\documentclass[journal]{IEEEtran}
\usepackage{cite}
\usepackage[cmex10]{amsmath}
\usepackage{amssymb}
\usepackage{amsthm}
\usepackage{mathrsfs}
\usepackage{bm}
\usepackage[mathscr]{eucal}
\usepackage{amssymb,amsmath,amsthm,amsfonts,latexsym}
\usepackage{amsmath,graphicx,bm,xcolor,url}
\usepackage{graphicx}
\usepackage{latexsym}
\usepackage{CJK}
\usepackage{indentfirst}
\usepackage{geometry}
\usepackage{psfrag}
\usepackage{setspace}
\usepackage{algorithmic}
\usepackage{algorithmic, cite}
\usepackage{algorithm}
\usepackage{array}
\usepackage{mdwmath}
\usepackage{mdwtab}
\usepackage{eqparbox}
\usepackage{url}
\usepackage{epstopdf}
\usepackage{epsfig,epsf,psfrag}
\usepackage{fixltx2e}
\usepackage{verbatim}
\usepackage{textcomp}
\usepackage{psfrag} 

\usepackage{caption}
\usepackage{subcaption}

\theoremstyle{plain}
\newtheorem{mythe}{Theorem}
\theoremstyle{plain}

\theoremstyle{plain}

\theoremstyle{plain}
\newtheorem{mypro}{Proposition}
\theoremstyle{plain}
\newtheorem{mycor}{Corollary}
\theoremstyle{plain}

\theoremstyle{plain}

\theoremstyle{plain}

\theoremstyle{plain}

\theoremstyle{plain}

\theoremstyle{plain}

\catcode`~=11 \def\UrlSpecials{\do\~{\kern -.15em\lower .7ex\hbox{~}\kern .04em}} \catcode`~=13

\allowdisplaybreaks[4]

\newcommand{\calC}{\mathcal{C}}
\newcommand{\calD}{\mathcal{D}}

\newcommand{\calL}{\mathcal{L}}

\newcommand{\calN}{\mathcal{N}}

\newcommand{\calR}{\mathcal{R}}

\newcommand{\ba}{\mathbf{a}}
\newcommand{\bA}{\mathbf{A}}
\newcommand{\bb}{\mathbf{b}}
\newcommand{\bB}{\mathbf{B}}

\newcommand{\bC}{\mathbf{C}}

\newcommand{\bE}{\mathbf{E}}

\newcommand{\bF}{\mathbf{F}}
\newcommand{\bg}{\mathbf{g}}
\newcommand{\bG}{\mathbf{G}}

\newcommand{\bH}{\mathbf{H}}
\newcommand{\bi}{\mathbf{i}}
\newcommand{\bI}{\mathbf{I}}

\newcommand{\boldm}{\mathbf{m}}
\newcommand{\bM}{\mathbf{M}}

\newcommand{\bR}{\mathbf{R}}

\newcommand{\bS}{\mathbf{S}}

\newcommand{\bT}{\mathbf{T}}
\newcommand{\bu}{\mathbf{u}}
\newcommand{\bU}{\mathbf{U}}
\newcommand{\bv}{\mathbf{v}}
\newcommand{\bV}{\mathbf{V}}
\newcommand{\bw}{\mathbf{w}}
\newcommand{\bW}{\mathbf{W}}
\newcommand{\bx}{\mathbf{x}}
\newcommand{\bX}{\mathbf{X}}
\newcommand{\by}{\mathbf{y}}
\newcommand{\bY}{\mathbf{Y}}
\newcommand{\bz}{\mathbf{z}}
\newcommand{\bZ}{\mathbf{Z}}

\newcommand{\bbC}{\mathbb{C}}

\newcommand{\bbE}{\mathbb{E}}

\newcommand{\bbR}{\mathbb{R}}

\DeclareMathAlphabet{\mathbsf}{OT1}{cmss}{bx}{n}
\DeclareMathAlphabet{\mathssf}{OT1}{cmss}{m}{sl}

\DeclareSymbolFont{bsfletters}{OT1}{cmss}{bx}{n}
\DeclareSymbolFont{ssfletters}{OT1}{cmss}{m}{n}
\DeclareMathSymbol{\bsfGamma}{0}{bsfletters}{'000}
\DeclareMathSymbol{\ssfGamma}{0}{ssfletters}{'000}
\DeclareMathSymbol{\bsfDelta}{0}{bsfletters}{'001}
\DeclareMathSymbol{\ssfDelta}{0}{ssfletters}{'001}
\DeclareMathSymbol{\bsfTheta}{0}{bsfletters}{'002}
\DeclareMathSymbol{\ssfTheta}{0}{ssfletters}{'002}
\DeclareMathSymbol{\bsfLambda}{0}{bsfletters}{'003}
\DeclareMathSymbol{\ssfLambda}{0}{ssfletters}{'003}
\DeclareMathSymbol{\bsfXi}{0}{bsfletters}{'004}
\DeclareMathSymbol{\ssfXi}{0}{ssfletters}{'004}
\DeclareMathSymbol{\bsfPi}{0}{bsfletters}{'005}
\DeclareMathSymbol{\ssfPi}{0}{ssfletters}{'005}
\DeclareMathSymbol{\bsfSigma}{0}{bsfletters}{'006}
\DeclareMathSymbol{\ssfSigma}{0}{ssfletters}{'006}
\DeclareMathSymbol{\bsfUpsilon}{0}{bsfletters}{'007}
\DeclareMathSymbol{\ssfUpsilon}{0}{ssfletters}{'007}
\DeclareMathSymbol{\bsfPhi}{0}{bsfletters}{'010}
\DeclareMathSymbol{\ssfPhi}{0}{ssfletters}{'010}
\DeclareMathSymbol{\bsfPsi}{0}{bsfletters}{'011}
\DeclareMathSymbol{\ssfPsi}{0}{ssfletters}{'011}
\DeclareMathSymbol{\bsfOmega}{0}{bsfletters}{'012}
\DeclareMathSymbol{\ssfOmega}{0}{ssfletters}{'012}

\newcommand{\tili}{\widetilde{i}}

\newcommand{\hatM}{\widehat{M}}

\newcommand{\tilM}{\widetilde{M}}

\newcommand{\hatbM}{\widehat{\bM}}

\newcommand{\tilv}{\widetilde{v}}

\newcommand{\tilbz}{\widetilde{\bz}}
\newcommand{\tilbZ}{\widetilde{\bZ}}

\newcommand{\bLambda}{\bm{\Lambda}}
\newcommand{\bSigma	}{\bm{\Sigma}}
\newcommand{\bPsi}{\bm{\Psi}}

\def\norm#1{\left\| #1 \right\|}
\def\norm2#1{\left\| #1 \right\|_2}
\def\norm22#1{\left\| #1 \right\|_2^2}

\DeclareMathOperator{\diag}{diag}

\DeclareMathOperator{\rank}{rank}

\newcommand{\qednew}{\nobreak \ifvmode \relax \else
      \ifdim\lastskip<1.5em \hskip-\lastskip
      \hskip1.5em plus0em minus0.5em \fi \nobreak
      \vrule height0.75em width0.5em depth0.25em\fi}

\usepackage{setspace}

\allowdisplaybreaks[1]
\usepackage{float}
\usepackage{cite}
\usepackage{stfloats}
\usepackage{multicol}

\usepackage[normalem]{ulem}

\DeclareMathOperator{\Tr}{Tr}

\title{Magnetic MIMO Signal Processing and Optimization for Wireless Power Transfer}
\author{Gang Yang, \IEEEmembership{Member, IEEE}, Mohammad R. Vedady Moghadam, \IEEEmembership{Member, IEEE}, and Rui Zhang, \IEEEmembership{Fellow, IEEE}
\thanks{Manuscript received August 21, 2016; revised December 27, 2016 and February 1, 2017; accepted February 9, 2017. The work of G. Yang was supported in part by the National Natural Science Foundation of China under Project 61601100. This work was presented in part at the IEEE International Conference on Acoustics,
Speech, and Signal Processing, Shanghai, China, March 20-25, 2016.}
\thanks{G.~Yang was with the Department of Electrical and Computer Engineering, National University of Singapore, Singapore 117583. He is now with the National Key Laboratory of Science and Technology on Communications, University of Electronic Science and Technology of China, Chengdu 611731, China (e-mail: yanggang@uestc.edu.cn).}
\thanks{M. R. V. Moghadam is with the Department of Electrical and Computer Engineering, National University of Singapore, Singapore 117583 (email: elemrvm@nus.edu.sg).}
\thanks{R. Zhang is with the Department of Electrical and Computer Engineering, National University of Singapore, Singapore 117583, and also with the Institute for Infocomm Research, Agency for Science, Technology and Research, Singapore 138632 (e-mail: elezhang@nus.edu.sg).}
\vspace{-5mm}}

\begin{document}
\begin{spacing}{1.0}
\maketitle
\begin{abstract}
In magnetic resonant coupling (MRC) enabled multiple-input multiple-output (MIMO) wireless power transfer (WPT) systems, multiple transmitters (TXs) each with one single coil are used to enhance the efficiency of simultaneous power transfer to multiple single-coil receivers (RXs) by constructively combining their induced magnetic fields at the RXs, a technique termed ``magnetic beamforming''.
\textcolor{black}{In this paper, we study the optimal magnetic beamforming design in a multi-user MIMO MRC-WPT system.}
\textcolor{black}{We introduce the multi-user power region that constitutes all the achievable power tuples for all RXs, subject to the given total power constraint over all TXs as well as their individual peak voltage and current constraints.
We characterize  each boundary point of the power region by maximizing the sum-power deliverable to all RXs subject to their  minimum harvested power constraints, which are proportionally set based on a given power-profile vector to ensure fairness.}
For the special case without the TX peak voltage and current constraints, we derive the optimal TX current allocation for the single-RX setup in closed-form as well as that for the multi-RX setup by applying the techniques of semidefinite relaxation (SDR) and time-sharing.
In general, the problem is a non-convex quadratically constrained quadratic programming  (QCQP), which is difficult to solve.
For the case of one single RX, we show that the SDR of the problem is tight, and thus the problem can be efficiently solved. For the general case with multiple RXs, based on SDR we obtain two approximate solutions by applying the techniques of time-sharing and randomization, respectively.
\textcolor{black}{Moreover, for practical implementation of magnetic beamforming, we propose a novel signal processing method to estimate the magnetic MIMO channel due to the mutual inductances  between TXs and RXs.}
Numerical results show that our proposed magnetic channel estimation and \textcolor{black}{adaptive} beamforming schemes are practically effective, and can significantly improve the power transfer efficiency and multi-user performance trade-off in MIMO MRC-WPT systems compared to the benchmark scheme of uncoordinated WPT with fixed  identical TX current.
\end{abstract}
\section{Introduction}\label{introduction}
\PARstart{N}{ear-field} wireless power transfer (WPT) has drawn significant interests recently due to its high efficiency for delivering power to electric loads without the need of any wire. Near-field WPT can be realized by inductive coupling (IC) for short-range applications within centimeters, or magnetic resonant coupling (MRC) for mid-range applications up to a couple of meters. Although short-range WPT has been in widely commercial use (e.g., electric toothbrushes), mid-range WPT is still largely under research and prototype.
In $2007$, a milestone experiment has successfully demonstrated that based on strongly coupled magnetic resonance, a single transmitter (TX) is able to transfer $60$ watts of power wirelessly with $40\%$--$50\%$ efficiency to a single receiver (RX) at a distance about $2$ meters. Motivated by this landmark experimental result, the research in MRC enabled WPT (MRC-WPT) has grown fast and substantially (see e.g., \cite{HuiLee14} and the references therein).

\textcolor{black}{MRC-WPT with generally multiple TXs and/or multiple RXs has been studied in the literature \cite{AhnHong13, LangSarris14, JadidianKatabi14, RezaZhang16TISPN, CasanovaLin09,KursMoffattSoljacic10}.
Under the multiple-input single-output (MISO) setup, \cite{AhnHong13} has studied an MRC-WPT system with two TXs and one single RX, while the analytical results proposed in this paper cannot be directly extended to the case with more than two TXs. In~\cite{LangSarris14}, a convex optimization problem has been formulated to maximize the efficiency of MISO MRC-WPT by jointly optimizing all TX currents together with the RX impedance.}
\textcolor{black}{However, the study in \cite{LangSarris14} has not considered the practical circuit constraints at individual TXs, such as peak voltage and current  constraints, and also its solution cannot be applied to the muti-RX setup.
Recently, \cite{JadidianKatabi14} has reported a wireless charger with an array of TX coils which can efficiently charge a mobile phone $40$cm away from the charging unit, regardless of the phone's orientation.}
\textcolor{black}{On the other hand, under the single-input multiple-output (SIMO) setup, an MRC-WPT system with one single TX and multiple RXs has been studied in \cite{RezaZhang16TISPN}, in which the load resistances of all RXs are jointly optimized to minimize the total transmit power drawn while achieving fair power delivery to the loads at different RXs,   even subject to their near-far distances to the TX.
For the general multiple-input  multiple-output (MIMO) setup, in \cite{CasanovaLin09} it has been experimentally demonstrated  that employing multiple TX coils can enhance the power delivery to multiple RXs simultaneously, in terms of both  efficiency and deliverable power.
However, this work has not addressed  how to design the system parameters to achieve optimal performance.}

Currently, there are two main industrial organizations on standardizing wireless charging, namely, the Wireless Power Consortium (WPC) which developed the ``Qi'' standard based on magnetic induction, and the Alliance for Wireless Power (A4WP) which developed the ``Rezence'' specification based on magnetic resonance.
\textcolor{black}{The Rezence specification advocates a superior charging range, the capability to charge multiple devices concurrently, and the use of two-way Bluetooth communication between the charger and devices for real-time charging control.
These features make Rezence a promising technology for high-performance wireless charging in future.}
However, in the current Rezence specification, one single TX coil is used in the power transmitting unit, i.e., only the SIMO MRC-WPT is considered.
Generally, deploying multiple TXs can help focusing their generated magnetic fields more efficiently toward one or more RXs simultaneously~\cite{JadidianKatabi14}, thus achieving a magnetic beamforming gain, in a manner analogous to multi-antenna beamforming in the far-field wireless information and/or power transfer based on electromagnetic (EM) wave radiation  \cite{GershmanSPM10, MIMOWIPTZhang13, ShenLiChangTSP14, ShiPengXuHongCaiTSP16}. It is worth noting that applying signal processing and optimization techniques for improving the efficiency of far-field WPT systems has recently drawn significant interests (see, e.g., the work on transmit beamforming design \cite{YangHoGuan13, SWIPTXuZhangTSP14}, channel acquisition method \cite{XuZhangTSP14, XuZhangTSP16}, waveform optimization  \cite{ClerckxBayguzinaTSP16}, and power scheduling policy for WPT networks~\cite{XiaAissaTSP15}). However, to our best knowledge, there has been no prior work on magnetic beamforming optimization under practical TX circuit constraints, for a MIMO MRC-WPT system with arbitrary numbers of TXs and RXs, which motivates \textcolor{black}{our work}. The results of this paper can be potentially applied in e.g., the Rezence specification for the support of multi-TX WPT \textcolor{black}{for} performance enhancement.

\textcolor{black}{In this paper, as shown in Fig. \ref{fig:Fig_app}, we consider a general MIMO MRC-WPT system with multiple RXs and multiple TXs where the TXs' source currents (or equivalently voltages) can be adjusted such that their induced magnetic fields are optimally combined at each of the RXs, to maximize the power delivered. }
We introduce the multi-user power region to characterize the optimal performance trade-offs among the RXs, which constitutes all the achievable power tuples deliverable to all RXs subject to the given total consumed power constraint over all TXs as well as practical peak voltage and current constraints at individual TXs.
\begin{figure}[t!]
\centering
\includegraphics[width=.9\columnwidth]{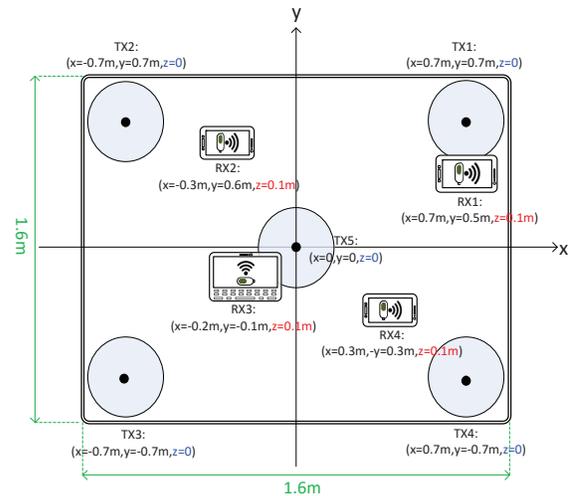}
\caption{\textcolor{black}{Example setup  of our considered MIMO MRC-WPT system: a rectangular table with five built-in wireless chargers attached below its surface and four receivers randomly  placed on it for wireless charging.}}
\label{fig:Fig_app} \vspace{-2mm}
\end{figure}

The main contributions of this paper are summarized as follows.
\begin{itemize}
\item \textcolor{black}{In order to characterize the optimal performance trade-offs among all RXs by finding all the boundary points of the multi-user power region, we apply the technique of power profile.
Specifically, we obtain each boundary point by maximizing the sum-power deliverable to all \textcolor{black}{RXs} subject to the minimum harvested power constraints at different RX loads which are proportionally set based on a given power-profile vector.
We propose an iterative algorithm to solve \textcolor{black}{this} problem, which requires to solve a TX sum-power minimization problem at each iteration to optimally allocate \textcolor{black}{the} TX currents.}
\item For the special case of one single RX, identical TX resistances and without the TX peak voltage and current constraints, we show that the optimal current at each TX should be proportional to the mutual inductance between its TX coil and the RX coil. This optimal magnetic beamforming design for MISO MRC-WPT system is analogous to the maximal-ratio-transmission (MRT) based beamforming in the far-field radiation-based  WPT~\cite{MIMOWIPTZhang13}.
\item  In general, the TX sum-power minimization  problem is \textcolor{black}{a} non-convex quadratically constrained quadratic programming (QCQP). For the case of one single RX, \textcolor{black}{with arbitrary TX resistances and the peak voltage and current constraints at individual TXs applied}, we show that the semidefinite relaxation (SDR) of the problem is tight, and thus the problem can be efficiently solved via the semidefinite  programming (SDP) by using existing optimization software such as CVX~\cite{CVXTool}.
For the general case with multiple RXs, based on SDR, we obtain two approximate solutions by applying the techniques of time-sharing and randomization, respectively. In particular, for the special case without the TX peak and voltage constraints, the time-sharing based solution is shown to be optimal.
\item \textcolor{black}{For practical implementation of magnetic bemaforming, it is essential to obtain the magnetic channel knowledge on the mutual inductance  between each pair of TX coil and RX coil. To this end, we propose a novel  magnetic MIMO channel estimation scheme, which is shown to be efficient and accurate \textcolor{black}{by simulations}. The channel estimation and feedback design for MIMO or multi-antenna based wireless communication systems has been extensively studied in the literature (see. e.g., \cite{DJLoveAndrewsJSAC08} and the references therein). However, it is shown in this paper that the magnetic MIMO channel estimation problem in MRC-WPT has a different   structure, which cannot be directly solved by existing methods in wireless communication.}
\item By extensive numerical results, we show that our proposed magnetic beamforming designs are practically effective, and can significantly  enhance the energy efficiency as well as  the multi-user performance trade-off in MIMO MRC-WPT, as compared to the benchmark scheme of uncoordinated WPT with fixed identical \textcolor{black}{current at all TXs}.
\end{itemize}

\begin{table*}[t!]
\centering
\textcolor{black}{\caption{\textcolor{black}{List of main variable notations and their meanings}} \label{table1}
\small{
\begin{tabular}{l | l}
  \hline \hline
   Notation & Meaning\\
  \hline
  $N, Q$ & Number of TXs and RXs, respectively \\
  $n, q$ & Index for TXs and RXs, respectively \\
  $w$   & Operating angular frequency \\
  $v_{\textrm{tx}, n}$ & Phasor representation for complex voltage of TX $n$ \\
  $i_{\textrm{tx}, n}, \bar{i}_{\text{tx},n}, \hat{i}_{\text{tx},n}$ & Phasor representation for complex current, real-part and imaginary-part of current of TX $n$, respectively \\
  $\bi$, $\bar{\bi}$, $\hat{\bi}$ & TX current vector $\bi=[i_{\textrm{tx}, 1}~ \ldots~i_{\textrm{tx}, N}]^T$, its real-part and imaginary-part, respectively \\  %
$i_{\textrm{rx}, q}, \bar{i}_{\textrm{rx}, q}, \hat{i}_{\textrm{rx}, q}$ & Phasor representation for complex current, real-part and imaginary-part of current of RX $q$, respectively \\
$L_{{\textrm{tx}},  n}$, $C_{{\textrm{tx}}, n}$ & Self-inductance and capacitance of the $n$-th TX coil, respectively \\
$L_{{\textrm{rx}}, q}$, $C_{{\textrm{rx}}, q}$ & Self-inductance and capacitance of the $q$-th RX coil, respectively \\
$r_{\textrm{tx}, n}$ & Total source resistance of the $n$-th TX \\
$\bR$ & Diagonal resistance matrix $\bR = \diag \{r_{\textrm{tx}, 1}, \ldots, r_{\textrm{tx}, N}\}$ \\
$r_{ \textrm{rx,p}, q}$, $r_{ \textrm{rx,l}, q}$, $r_{\textrm{rx}, q}$ & Parasitic resistance, load resistance and total resistance of RX $q$, respectively \\
$M_{nq}$, $\tilM_{n k}$ & Mutual inductance between TX $n$ and RX $q$ / TX $k$ with $k \neq n$, respectively \\
$\boldm_q$ & Vector of mutual inductance between RX $q$ and all TXs \\
$\bM_q$ & Rank-one matrix $\bM_q=\boldm_q \boldm_q^T$ for RX $q$ \\
$\bB, \overline{\bB}, {\widehat{\bB}}$ & Impedance matrix, its real-part and imaginary-part, respectively \\ 
$\bB_n$ & Rank-one matrix $\bB_n = \bb_n \bb_n^H$, with $\bb_n$ denoting the $n$-th column of $\bB$ \\
$p_{\sf{tx}}$ & Total power drawn from all TXs \\
$p_{\textrm{rx}, q}$ & Power delivered to the load of RX $q$ \\
$P_T$ & Maximum total power drawn by all TXs \\
$V_n$, $A_{n}$ & Maximum amplitude of voltage and current of TX $n$, respectively \\
$\bm{\alpha}$ & Power-profile vector \\
$\bW_n$ & Rank-one matrix with the $n$-th diagonal element being one and others zero \\
$P$ & Sum-power delivered to all RXs \\
$\bX$ & Rank-one matrix $\bX=\bi \bi^H$ \\
$L$ & Rank of optimal SDR solution $\bX^{\star}$ \\
$\bV$ & Singular matrix of $\bX^{\star}$, $\bV=[\bv_1 \; \ldots \; \bv_L]$ \\ 
$\bLambda$ & $L$-order diagonal matrix with diagonal elements given by eigenvalues of $\bX^{\star}$ \\ 
$\tau_l$ & Transmission time of the $l$-th WPT slot in time-sharing based solution \\
$e_{\textrm{rx}, q, t}$ & Error of the $q$-th RX's current in the $t$-th channel-training slot \\
  \hline
\end{tabular}
}
}
\end{table*}

The rest of this paper is organized as follows. Section~\ref{system_model} introduces the system model for MIMO MRC-WPT. Section~\ref{sec: formulation} presents the problem formulation to characterize the boundary points of the multi-user power region. Section~\ref{sec: solution} presents the optimal and approximate solutions for the formulated problem under various setups. Section~\ref{sec: CE} presents the algorithms for magnetic MIMO channel estimation. Section~\ref{sec: simulation} provides the numerical results. Section~\ref{sec:conslusion} concludes the paper.

\textcolor{black}{The notations for main variables used in this paper are listed in Table~\ref{table1} for the ease of reading. Moreover,} we use the following math notations in this paper. $|\cdot|$ means the operation of taking the absolute value. $\bX \succcurlyeq 0$ means that the matrix $\bX$ is positive semidefinite (PSD). $\mathrm{Re} \{\cdot\}$ means the operation of taking the real part. $\Tr(\cdot)$ means the trace operation. $\bigcup$ is the union operation of sets. $\bbE[\cdot]$ denotes the statistical expectation. \textcolor{black}{$\bv \sim \calC \calN(\bm{\mu}, \bC)$} means that the random vector $\bv$ follows the circularly symmetric complex Gaussian (CSCG) distribution with mean vector \textcolor{black}{$\bm{\mu}$} and covariance matrix \textcolor{black}{$\bC$}. The $(\cdot)^T, (\cdot)^{\ast}$ and $(\cdot)^H$ represent the transpose, conjugate, and conjugate transpose \textcolor{black}{operations}, respectively.

\section{System Model} \label{system_model}
As shown in Fig.~\ref{fig:Fig1}, we consider a MIMO MRC-WPT system with $N \geq 1$ TXs each equipped with a single coil, and $Q \geq 1$ single-coil RXs.
\textcolor{black}{We assume that the RXs are all legitimate users  for wireless charging.}
Each TX $n$, $n=1, \ldots, N$,  is connected to a stable power source supplying sinusoidal voltage over time given by $\tilv_{\textrm{tx}, n}(t) = {\mathrm{Re}} \{v_{\textrm{tx}, n} e^{jwt}\}$, with $v_{\textrm{tx}, n}$ denoting the complex voltage and $w > 0$ denoting the operating angular frequency. Let $\tili_{\textrm{tx}, n}(t) = {\mathrm{Re}} \{i_{\textrm{tx}, n} e^{jwt}\}$ denote the steady-state current flowing through TX $n$, with the complex current $i_{\textrm{tx}, n}$. The current produces a time-varying magnetic flux in the $n$-th TX coil, which passes through the coils of all RXs and induces time-varying currents in them. Let $\tili_{q}(t) = {\mathrm{Re}} \{i_{\textrm{rx}, q} e^{jwt}\}$ denote the steady-state current in the $q$-th RX coil, $q=1, \ldots, Q$, with the complex current $i_{\textrm{rx}, q}$.


Let $M_{nq}$ and $\tilM_{n k}$ denote the mutual inductance between the $n$-th TX coil and the $q$-th RX coil, and the mutual inductance between the $n$-th TX coil and the $k$-th TX coil with $k \neq n$, respectively. The mutual inductance is a real number, either positive or negative, which depends on the physical characteristics of each pair of TX and RX coils such as their relative distance, orientations, etc. \cite{JadidianKatabi14}.\footnote{\textcolor{black}{In this paper, the values of  mutual inductances (i.e., magnetic channels) are assumed to be purely real, since our considered MRC-WPT system operates under the near-field condition for which EM wave radiation is negligible and hence the imaginary-part of each inductance value can be set as zero.}}
\textcolor{black}{Specifically, the negative sign of mutual inductance $M_{nq}$ ($\tilM_{n k}$) indicates that the current induced at the coil of RX  $q$ (TX $k$) due to the current flowing at the coil of TX $n$ is in the opposite of the reference direction assumed  (as shown in Fig. \ref{fig:Fig1}, the reference current direction at each TX/RX is set to be clockwise in this paper for convenience).}
In this paper, we assume that the mutual coupling between any pair of RX coils is negligible, \textcolor{black}{as shown in Table  \ref{table_inductance_new} later  for our considered numerical example}, due to their small sizes in practice and the assumption  that they are well  separated from each other.
\begin{figure} [t!]
\centering
\includegraphics[width=.85\columnwidth]{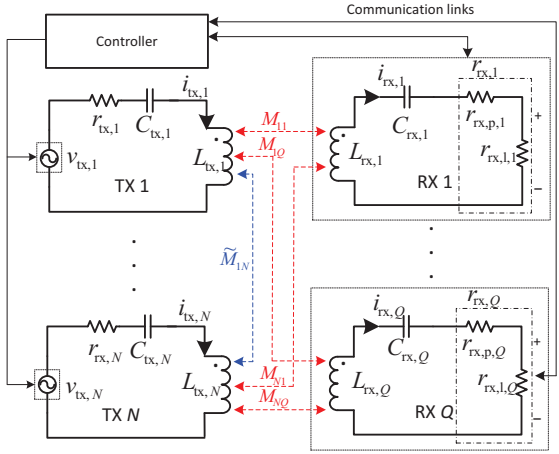}
\caption{System model of MIMO MRC-WPT.}
\label{fig:Fig1}
\vspace{-0.5cm}
\end{figure}

We denote the self-inductance and the capacitance of the $n$-th TX coil ($q$-th RX coil) by $L_{{\textrm{tx}},  n} >0$ ($L_{{\textrm{rx}}, q} > 0$) and $C_{{\textrm{tx}}, n} >0$ ($C_{{\textrm{rx}}, q} > 0$), respectively. The capacitance values are set as $C_{{\textrm{tx}}, n} = \frac{1}{L_{\textrm{tx}, n} w^2}$ and $C_{\textrm{rx}, q} = \frac{1}{L_{\textrm{rx}, q} w^2}$, such that all TXs and RXs have the same resonant angular frequency, $w$. Let $r_{\textrm{tx}, n} >0$ denote the total source resistance of the $n$-th TX. Define the diagonal resistance matrix as $\bR \triangleq \diag \{r_{\textrm{tx}, 1}, \ldots, r_{\textrm{tx}, N}\}$. The resistance of each RX $q$, denoted by $r_{\textrm{rx}, q}$, consists of the parasitic resistance $r_{ \textrm{rx,p}, q} >0$ and the load resistance $r_{ \textrm{rx,l}, q} >0$, i.e., $r_{\textrm{rx}, q} = r_{ \textrm{rx,p}, q} + r_{\textrm{rx,l}, q}$.
The load is assumed to be purely resistive.
\textcolor{black}{It is also assumed that the load resistance is sufficiently larger than the parasitic resistance at each RX $q$ such that $r_{ \textrm{rx,l}, q}/r_{\textrm{rx}, q} \approx 1$. This is practically required to ensure that most of the  energy harvested by the coil at each RX can be delivered to its load.}

\textcolor{black}{In our considered MRC-WPT system, we assume that there is a controller installed which can communicate with all TXs and RXs (e.g., using Bluetooth as in the Rezence specification) such that it can collect the information of all system parameters (e.g., RX loads and currents) required to design and implement magnetic beamforming.
We also assume that the RXs all have sufficient  initial energy stored in their batteries, which enables them to conduct the necessary current measurement and  send relevant  information to the central controller to implement magnetic beamforming.}
\textcolor{black}{However, for simplicity, we ignore the energy consumed for such operations at RXs.}
Last, for convenience, we treat the complex TX currents $i_{\textrm{tx}, n}$'s as design variables,\footnote{\textcolor{black}{In practice, it may be more convenient to use voltage source instead of current source. Therefore, after designing the TX currents $i_{\textrm{tx}, n}$'s, the corresponding voltages $v_{\textrm{tx}, n}$'s can be computed and set by the controller accordingly (see \eqref{eq:voltage_transmitter_n} and \eqref{eq:voltage_transmitter_n_vector}). Moreover, in the case of adjustable voltage sources, impedance matching  can be conducted in series with the sources, each of which can be adjusted in real time to match the current flowing in its corresponding TX to the optimal value obtained by magnetic beamforming design.}} which can be adjusted by the controller in real time to realize \textcolor{black}{adaptive} magnetic beamforming.

By applying Kirchhoff's circuit law to the $q$-th RX, we obtain its current $i_{\textrm{rx}, q}$ as
\begin{align}
 i_{\textrm{rx}, q} &= \frac{jw}{r_{\textrm{rx}, q}} \sum \limits_{n=1}^N M_{nq} i_{\textrm{tx}, n}. \label{eq:current_receiver_q}
\end{align}

\noindent Denote the vector of all TX currents as $\bi=[i_{\textrm{tx}, 1}~ \ldots~i_{\textrm{tx}, N}]^T$. Moreover,  denote the vector of mutual inductances between the $q$-th RX coil and all TX coils as \textcolor{black}{$\boldm_q = [M_{1q} \; \ldots \; M_{Nq}]^T$}, and define the rank-one matrix $\bM_q \triangleq \boldm_q \boldm_q^T$. From \eqref{eq:current_receiver_q}, the power delivered to the load of the $q$-th RX is
\begin{align}
  p_{\textrm{rx}, q} &= \frac{1}{2} |i_{\textrm{rx}, q}|^2 r_{ \textrm{rx}, q} =\frac{w^2}{2 r_{\textrm{rx}, q}} {\bi}^H \bM_q {\bi}.
  \label{eq:load_power_vec_complex}
\end{align}
Similarly, by applying Kirchhoff's circuit law to each TX $n$, we obtain its source voltage as
\begin{align}
v_{\textrm{tx}, n}
&= \left( r_{\textrm{tx}, n} +   \sum \limits_{q=1}^Q \frac{M_{nq}^2 w^2}{r_{\textrm{rx}, q}} \right) i_{\textrm{tx}, n} + \nonumber \\
&\qquad \sum \limits_{k \neq n}   \left( j w \tilM_{nk} +   \sum \limits_{q=1}^Q \frac{M_{nq} M_{qk} w^2}{r_{\textrm{rx}, q}}  \right) i_{\textrm{tx}, k}. \label{eq:voltage_transmitter_n}
\end{align}

\textcolor{black}{Next, we derive the total power drawn from all TXs in terms of the vector of TX currents ${\bi}$.}  Let us define an $N \times N$ impedance matrix $\bB $ as
\begin{align}
\bB  &= {\overline{\bB}}  + j {\widehat{\bB}},\label{eq_complex_B}
\end{align}
where the elements in ${\overline{\bB}}$ and $\widehat{\bB}$ are respectively given by
\begin{align}
{\overline{B}}_{nk}  &=
\left\{ \begin{array}{cl}
r_{\textrm{tx}, n} +   \sum \limits_{q=1}^Q \frac{M_{nq}^2 w^2}{r_{\textrm{rx}, q}} , &\mbox{if}\; k=n\\
\sum \limits_{q=1}^Q \frac{M_{nq} M_{qk} w^2}{r_{\textrm{rx}, q}}, \qquad &\mbox{otherwise;} \\
\end{array}
\right. \label{eq_complex_B_real}\\
{\widehat{B}}_{nk} &= \left\{ \begin{array}{cl}
0, &\mbox{if}\; k=n\\
- w \tilM_{ nk}, \qquad &\mbox{otherwise.} \\
\end{array}
\right.\label{eq_complex_B_imaginar}  \end{align}
Note that the matrices $\bB , \; {\overline{\bB}} $ and ${\widehat{\bB}}$ are all symmetric, since $M_{ nk} = M_{ kn}, \; \forall n \neq k$. Denote the $n$-th column of the matrices $\bB , \; \overline{\bB} , \; \widehat{\bB}$ by $\bb_n ,\; \overline{\bb}_n , \; \widehat{\bb}_n$, respectively. We also define the rank-one matrices $\bB_n \triangleq \bb_n \bb_n^H, n=1, \ldots, N$. It can be shown that both $\overline{\bB}$ and $\bB_n$'s are PSD matrices.
The matrix ${\overline{\bB}} $ can be also rewritten as
\begin{align}
{\overline{\bB}} = \bR + w^2 \sum \limits_{q=1}^Q \frac{\bM_q}{r_{\textrm{rx}, q}}. \label{eq:barB}
\end{align}
Accordingly, the source voltage of each TX $n$ given in  \eqref{eq:voltage_transmitter_n} can be equivalently re-expressed as
\begin{align}
  v_{\textrm{tx}, n}= {\bb}_n^H  \bi. \label{eq:voltage_transmitter_n_vector}
\end{align}
From \eqref{eq_complex_B} and \eqref{eq:voltage_transmitter_n_vector}, the total power drawn from all TXs is given by
\begin{align}
p_{\sf{tx}} &=\frac{1}{2} {\mathrm{Re}} \left\{ \sum \limits_{n=1}^N \bi^H {\bb}_n  i_{ n}\right\}
= \frac{1}{2}  {{\bi}}^H  {\overline{\bB}}  {{\bi}}. \label{eq:p_transmit_general_vec_complex}
\end{align}

Note that from \eqref{eq_complex_B_real}, it follows that $p_{\sf tx}$ in \eqref{eq:p_transmit_general_vec_complex} in general depends on the mutual inductances $M_{nq}$'s between all TXs and RXs, but does not depend on the mutual inductances $\tilM_{nk}$'s among the TXs.
\section{Problem Formulation}\label{sec: formulation}
In this section, we first introduce the multi-user power region to characterize the optimal performance trade-offs among all RXs in \textcolor{black}{a} MIMO MRC-WPT system \textcolor{black}{introduced} in Section~\ref{sec:region_define}. Then, we formulate an optimization problem to find each boundary point of the power region corresponding to a given ``power-profile'' vector.
\subsection{Multi-user Power Region }\label{sec:region_define}
In this subsection, we define the multi-user power region under practical circuit constraints at TXs.
In particular, the power region consists of all the achievable power tuples that can be received by all RXs subject to the following constraints: the total power drawn by all TXs needs to be no larger than a given maximum power $P_T$, i.e., $p_{\textrm{tx}} \leq P_{T}$; the peak amplitude of the voltage $v_{\textrm{tx}, n}$ (current $i_{\textrm{tx}, n}$ ) at each TX $n$ needs to be \textcolor{black}{no larger} than a  given threshold $V_n$ ($A_{n}$), i.e., $|v_{\textrm{tx}, n}| \leq V_{n}, \ |i_{\textrm{tx}, n}| \leq A_n, \forall n=1, \ldots,N$.
\textcolor{black}{In this case, it can be easily verified  that the maximum transmit power at each TX $n$ is indeed capped by $\frac{1}{2} V_n A_n$.  Accordingly, to avoid the trivial case that the constraint $p_{\sf{tx}} \leq P_T$ is never active, we consider that $\sum \nolimits_{n=1}^N \frac{1}{2} V_n A_n > P_T$ holds in this paper.} The power region is thus formally defined as
\begin{align}\label{eq:power_region}
  \calR &\triangleq \underset{\substack{p_{\textrm{tx}} \leq P_T, \ |v_{\textrm{tx}, n}| \leq V_n, \\ |i_{\textrm{tx},n}| \leq A_n, \ n=1,  \ldots,  N}}{\bigcup} \; (p_{\textrm{rx}, 1}, \ p_{\textrm{rx}, 2}, \ \ldots, \ p_{\textrm{rx}, Q}),
\end{align}
where $p_{\textrm{rx}, q}, \ v_{\textrm{tx}, n}, \ p_{\textrm{tx}}$ are given in \eqref{eq:load_power_vec_complex}, \eqref{eq:voltage_transmitter_n_vector}, and \eqref{eq:p_transmit_general_vec_complex}, respectively. Note that the union operation in~\eqref{eq:power_region} has considered the possibility that some power tuples may be achievable only through ``time-sharing (TS)'' of a certain set of achievable power tuples each corresponding to a different set of feasible $v_{\textrm{tx}, n}$'s and $i_{\textrm{tx},n}$'s.

Next, we apply the technique of power-profile vector~\cite{RezaZhang16TISPN} to characterize all the boundary points of the power region, where each boundary power tuple corresponds to a Pareto-optimal performance trade-off among the RXs. Let $P$ denote the sum-power delivered to all RXs, i.e., $P=\sum \nolimits_{q=1}^Q p_{\textrm{rx}, q}$. Accordingly, we set  $p_{\textrm{rx}, q}=\alpha_q P$, where the coefficients $\alpha_q$'s are subject to $\sum \nolimits_{q=1}^Q \alpha_q = 1$ and $\alpha_q \geq 0, \ \forall q$. The vector $\bm{\alpha}=[\alpha_1 \ \alpha_2 \ \ldots \ \alpha_Q]^T$ is a given power-profile vector that specifies the proportion of the sum-power delivered to each RX $q$.
With each given $\bm{\alpha}$, the maximum achievable sum-power $P$ thus corresponds to a boundary point of the power region; Fig.~\ref{fig:Fig1A} illustrates the characterization of the power region boundary via the power profile technique for the case of $Q=2$ RXs.
\subsection{Optimization Problem}\label{sec:formulation}
In this subsection, we formulate an optimization problem to find different boundary points of the power region. Denote the $N$-dimensional complex space by $\bbC^N$, and let $\bW_n$ denote the rank-one matrix with the $n$-th diagonal element being one and all other elements being zero.
\begin{figure} [t!]
	\centering	\includegraphics[width=.7\columnwidth]{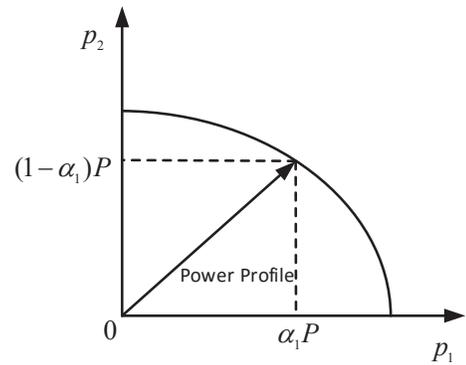}
	\caption{Illustration of characterization of power region boundary via the technique of power profile in a two-user \textcolor{black}{case}.} \label{fig:Fig1A}
\vspace{-0.2cm}
\end{figure}

From the definition in \eqref{eq:power_region}, each boundary point of the power region $\calR$ can be obtained by solving the following RX sum-power maximization problem with a given power-profile vector $\bm{\alpha}$ (for the case when TS is not required to achieve the boundary point of the multi-user power region corresponding to the given power profile $\bm{\alpha}$; see Proposition \ref{sec:TSopt_wo_peak_cons} in Section~\ref{sec: solution} for the case when TS is required),
%
\begin{subequations}\label{eq:optimP0}
\begin{align}
\mathrm{(P0)}: \  \ \underset{ {\bi} \in \bbC^N }{\text{max}} \ \ &P
\label{eq:rewardP0} \\
\text{s.t.} \ \
&\frac{w^2}{2 r_q} {\bi}^H \bM_q {\bi} \geq \alpha_q P, \;  q=1, \ldots, Q \label{eq:const1P0} \\
&\bi^H \bB_n \bi \leq V_{ n}^2, \;  n=1,\ldots, N \label{eq:const2P0} \\
&\bi^H \bW_n \bi \leq A_{ n}^2, \; n=1,\ldots, N \label{eq:const3P0} \\
&\frac{1}{2}   {{\bi}}^H  {\overline{\bB}} {{\bi}}  \leq P_T, \label{eq:const4P0}
\end{align}
\end{subequations}
where the inequalities \eqref{eq:const1P0}, \eqref{eq:const2P0} and \eqref{eq:const4P0} are due to \eqref{eq:load_power_vec_complex}, \eqref{eq:voltage_transmitter_n_vector}, and \eqref{eq:p_transmit_general_vec_complex}, respectively.
Given a power-profile vector $\bm{\alpha}$, $\mathrm{(P0)}$ can be solved by a bisection search over $P$, where in each search iteration, it suffices to solve a feasibility problem that checks whether all constraints of $\mathrm{(P0)}$ can be satisfied for some given $P$.
The converged optimal value of $P$ is denoted by $P^{\star}$.

The feasibility problem can be equivalently solved by first obtaining the minimum sum-power drawn from all TXs  by solving the following problem, denoted by $p_{\textrm{tx}}^{\star}$, and then comparing it with the given total power constraint for all TXs, $P_T$.  \textcolor{black}{Specifically, the TX  sum-power minimization problem is given by}
\begin{subequations}\label{eq:optimP1}
\begin{align}
\mathrm{(P1)}: \  \  \underset{ {\bi} \in \bbC^N }{\text{min}} \ \
&p_{\textrm{tx}} = \frac{1}{2}   {{\bi}}^H  {\overline{\bB}} {{\bi}} \label{eq:rewardP1} \\
\text{s.t.} \ \
&\frac{w^2}{2 r_q}  {\bi}^H \bM_q {\bi} \geq \alpha_q P, \; q=1, \ldots, Q \label{eq:const1P1} \\
&\bi^H \bB_n \bi \leq V_{ n}^2, \;  n=1, \ldots, N \label{eq:const2P1} \\
&\bi^H \bW_n \bi \leq A_{ n}^2, \;  n=1, \ldots, N. \label{eq:const3P1}
\end{align}
\end{subequations}

To summarize, the overall algorithm for solving (P0) is given in Algorithm \ref{Algorithm}.
Note that in the rest of this paper, we focus on solving problem $\mathrm{(P1)}$. However, $\mathrm{(P1)}$ is in general a non-convex QCQP  problem \cite{ConvecOptBoyd04} due to the constraints in \eqref{eq:const1P1}. Although solving non-convex QCQPs is difficult in general \cite{LuoZhangSPM10}, we study the optimal and approximate solutions to $\mathrm{(P1)}$ under various setups in Section~\ref{sec: solution}. Notice that for solving $\mathrm{(P1)}$, it is essential for the controller to have the knowledge of the mutual inductance values between any pair of TX coils as well as any pair of TX and RX coils. In practice, the TX-TX mutual inductance is constant with fixed TX positions and thus can be measured offline and stored in the controller.
However, due to the mobility of RXs (such as phones, tablets), the TX-RX mutual inductance is time-varying in general and thus needs to be estimated periodically. The magnetic channel estimation problem will be addressed \textcolor{black}{later} in Section \ref{sec: CE}.
\begin{algorithm}[t!] \small
	\caption{: Algorithm for $\mathrm{(P0)}$} \label{Algorithm}
	\begin{algorithmic}[1]
		\STATE Initialization: $P_{\min}=0$,   \textcolor{black}{$P_{\max}=P_T$}, and a small positive number $\epsilon$ (\textcolor{black}{$\epsilon=10^{-2}$ is set in our simulations}). \\
		\WHILE{$P_{\max} - P_{\min} > \epsilon$}
		\STATE $P=(P_{\min}+P_{\max})/2$.
		\IF{$\mathrm{(P1)}$ is not feasible}
		\STATE Go to step 8.
		\ELSIF{$p_{\textrm{tx}}^{\star} (P) > P_T$}
		\STATE Obtain the optimal solutions as $\bi^{\star}(P)$. 
		\STATE $P_{\max} \leftarrow P$.
		\ELSE
		\STATE Obtain the optimal solutions as $\bi^{\star}(P)$.
		\STATE $P_{\min} \leftarrow P$.
		\ENDIF
		\ENDWHILE
		\RETURN the optimal value and solution of $\mathrm{(P0)}$ as $P^{\star}=P$ and $\bi^{\star}=\bi^{\star}(P^{\star})$, respectively.
	\end{algorithmic}
\end{algorithm}

Last,  note that an alternative  approach to characterize the boundary of the multi-user power region is to solve a sequence of weighted sum-power maximization (WSPMax) problems for the RXs.
Compared to the TX sum-power minimization problem $\mathrm{(P1)}$ with the given RX minimum load power constraints, the WSPMax problem with the given maximum total TX power can be considered as its ``dual'' problem. \textcolor{black}{In practice, how to select weights in WSPMax so as to satisfy the minimum load power requirement at each RX is challenging.} Hence, in this paper, we study $\mathrm{(P1)}$ due to its practical usefulness in satisfying any given RX load power requirements.
\section{Solutions to Problem ($\mathrm{P1}$)}\label{sec: solution}
In this section, we first present the optimal solution to ($\mathrm{P1}$) for the special case without TX peak voltage and current constraints~\eqref{eq:const2P1} and~\eqref{eq:const3P1}, and then study the solution to ($\mathrm{P1}$) for the general case with all constraints.
\vspace{-1mm}
\subsection{Optimal Solution to (${P1}$) without Peak Voltage and Current Constraints} \label{sec:solution_wo_peak_cons}
In this subsection, we consider ($\mathrm{P1}$) for the ideal case without the TX peak voltage and current constraints given in \eqref{eq:const2P1} and \eqref{eq:const3P1}, respectively, to obtain useful insights and the performance limit of magnetic beamforming.

Denote the $N$-dimensional real space by $\bbR^N$. Let $\bi=\bar{\bi} + j \hat{\bi}$, where $\bar{\bi}, \ \hat{\bi} \in \bbR^N$. It is then observed that the real-part $\bar{\bi}$ and the imaginary-part $\hat{\bi}$ contribute in the same way to the total TX power in \eqref{eq:rewardP1} as well as the delivered load power in \eqref{eq:const1P1}, since both $\overline{\bB}$ and $\bM_q$'s are symmetric matrices. As a result, we can set $\hat{\bi}=\mathbf{0}$ without loss of generality and adjust $\bar{\bi}$ only, i.e., we need to solve
%
\begin{subequations}\label{eq:optimP2}
\begin{align}
\mathrm{(P2)}: \ \ \underset{ \bar{\bi}  \in \bbR^N}{\text{min}} \ \
& \frac{1}{2}  {\bar{\bi}}^T  {\overline{\bB}} {\bar{\bi}}   \label{eq:rewardP2} \\
\text{s.t.} \ \
&\frac{w^2}{2 r_{\textrm{rx}, q}}  {\bar{\bi}}^T \bM_q {\bar{\bi}} \geq \alpha_q P, \; q=1, \ldots, Q. \label{eq:const1P2}
\end{align}
\end{subequations}
Denote the space of $N$-order real matrices by $\bbR^{N \times N}$. Let $\bX = {\bar{\bi}} {\bar{\bi}}^T$. The SDR of $\mathrm{(P2)}$ is thus given by
%
\begin{subequations}\label{eq:optimP2-SDR}
\begin{align}
(\mathrm{P2}\mathrm{-SDR}): \ \ \underset{ \bX \in \bbR^{N \times N}}{\text{min}} \ \
&\frac{1}{2}   \Tr \left( {\overline{\bB}}  \bX \right)  \\
\text{s.t.} \ \ & \Tr \left( {\bM_q} \bX \right) \geq \frac{2 r_{\textrm{rx}, q} \alpha_q P}{w^2},        \nonumber \\
         &\qquad \quad  q=1, 2, \ldots, Q  \label{eq:const1P2SDP} \\
&\ \bX \succcurlyeq 0. \label{eq:const2P2SDP}
\end{align}
\end{subequations}

In general, $(\mathrm{P2}\mathrm{-SDR})$ is a convex relaxation of $(\mathrm{P2})$ by dropping the rank-one constraint on $\bX$. This relaxation is tight, if and only if the solution obtained for $(\mathrm{P2}\mathrm{-SDR})$, denoted by $\bX^{\star}$, is of rank one. In the following, we discuss the solutions to $(\mathrm{P2}\mathrm{-SDR})$ as well as that for $(\mathrm{P2})$ for the two cases with one single RX and multiple RXs, respectively.
\subsubsection{Single-RX Case}
Let $\bI_N$ denote the $N$-order identity matrix. For the case of single RX (i.e., RX $1$ with $\alpha_1=1$), the optimal solution to $\mathrm{(P2)}$ is obtained in closed-form as follows.
\begin{mythe}\label{the:solution_specialP2_Q1}
For the case of $Q=1$, the optimal solution to $\mathrm{(P2)}$ is ${\bar{\bi}}^{\star} = \beta \bu_1$, where $\beta$ is a constant such that the constraint \eqref{eq:const1P2} holds with equality, and $\bu_1$ is the eigenvector associated with the minimum eigenvalue, denoted by $\psi_1$, of the matrix
\begin{align}
  \bT=\bR + \frac{w^2 (1 -v^{\star})}{r_{\emph{\textrm{rx}}, 1}} \bM_1, \label{eq:solution_wo_constraint}
\end{align}
where $v^{\star}$ is chosen such that $\psi_1 =0$.
Particularly, for the case of identical TX resistances,  i.e., $\bR=r \bI_N$ with $r>0$, the optimal solution to ($\mathrm{P2}$) is simplified to
 \begin{align}
   {\bar{\bi}}^{\star} = \frac{\beta \boldm_1}{\| \boldm_1 \|_2}\label{eq:solution_wo_constraint_common_r}.
\end{align}
\end{mythe}

\begin{proof}
Please refer to Appendix~\ref{sec:proof_lemma2}.
\end{proof}

Theorem \ref{the:solution_specialP2_Q1} implies that for the case of single RX and identical TX resistances, the optimal current of each TX $n$ is proportional to the mutual inductance $M_{n1}$ between the TX $n$ and RX $1$.  \textcolor{black}{This is analogous to the maximal-ratio-transmission (MRT) based beamforming in the far-field wireless communication~\cite{GershmanSPM10}. However, magnetic beamforming operates in the near-field and thus the phase of each TX current only needs to take the  value of $0$ or $\pi$, i.e., the current is a positive or negative real number depending on its positive or negative mutual inductance with the RX, while in wireless communication  beamforming operates over the far-field, and as a result, the beamforming weight at each transmit antenna needs to be of the opposite phase of that of the wireless channel, which can be an arbitrary value within $0$ and $2\pi$.}
\subsubsection{Multiple-RX Case}\label{sec:MU_woPeak}
For the general case of multiple RXs, $(\mathrm{P2}\mathrm{-SDR})$ is a separable SDP with $Q$ constraints. We directly obtain the following result from~\cite[Thm. 3.2]{HuangPalomar10}.
\begin{mypro} \label{sec:rank_wo_peak_cons}
For the case of $Q \geq 1$, the rank of the optimal solution to $(\mathrm{P2}\mathrm{-SDR})$ is upper-bounded by
  \begin{align}
  \rank  \left( \bX^{\star} \right) \leq \sqrt{Q}. \label{eq:bound_wo_peak_cons}
\end{align}
\end{mypro}
From Lemma \ref{sec:rank_wo_peak_cons}, we have the following corollary.
\begin{mycor}\label{cor:solution_specialP2_Q123}
For $Q \leq 3$, the SDR in $(\mathrm{P2}\mathrm{-SDR})$ is tight, i.e., the optimal solution $\bX^{\star}$ to $(\mathrm{P2}\mathrm{-SDR})$ is always rank-one, which is given by $\bX^{\star} = {\bar{\bi}}^{\star} \left({\bar{\bi}}^{\star}\right)^T$. The optimal solution to $(\mathrm{P2})$ is thus ${\bar{\bi}}^{\star}$.
\end{mycor}
Note that for $Q \geq 4$, the optimal solution of $\bX^{\star}$ to $(\mathrm{P2}\mathrm{-SDR})$ may have a rank higher than 1, which is thus not feasible to $(\mathrm{P2})$. In general, $(\mathrm{P2}\mathrm{-SDR})$ can be efficiently solved by existing software such as CVX~\cite{CVXTool}.


In the following, we propose a time-sharing (TS) based scheme to achieve the same optimal value of problem $(\mathrm{P2}\mathrm{-SDR})$. Let $L$ be the rank of the obtained solution $\bX^{\star}$ for $(\mathrm{P2}\mathrm{-SDR})$, i.e., $L=\rank  \left( \bX^{\star} \right)$, with $L \leq N$. Denote the singular-value-decomposition (SVD) of $\bX^{\star}$ by $\bX^{\star}=\bV \bLambda \bV^H$, where $\bV=[\bv_1 \; \ldots \; \bv_L]$ is an $N \times L$ matrix with $\bV^H \bV=\bI_L$ and $\bLambda \triangleq \diag \{\lambda_1,\ldots,  \lambda_L\}$ is an $L$-order diagonal matrix with the diagonal elements given by $\lambda_1 \geq  \lambda_2 \ge \ldots  \lambda_L >0$.


\textcolor{black}{To perform magnetic beamforming in a TS manner, we divide WPT  into $L$ orthogonal time slots, indexed by $l \in \{1,\ldots,L\}$, where  slot $l$ takes a portion of the total transmission time given by $\tau_l$, with $0<\tau_l<1$ and $\sum_{l=1}^L \tau_l=1$}. In particular, we set
\begin{align}
\tau_l &= \frac{\lambda_l}{\sum_{k=1}^L \lambda_k}.\label{eq:time_TS_woPeak}
\end{align}
In the $l$-th slot, the TX current vector is then given by
\begin{align}
{\bar{\bi}_l}^{\star} &= \sqrt{\sum \limits_{k=1}^L \lambda_k} \; \bv_l. \label{eq:current_TS_woPeak}
\end{align}
We have the following result on the TS scheme.
\begin{mypro} \label{sec:TSopt_wo_peak_cons}
For the case without peak voltage and current constraints, the TS scheme given in~\eqref{eq:time_TS_woPeak} and~\eqref{eq:current_TS_woPeak} achieves the same optimal value of $(\mathrm{P2}\mathrm{-SDR})$.
\end{mypro}

\begin{proof}
With the TS scheme, the total delivered power to each RX $q$ over $L$ time slots is
\begin{align}
  \sum \limits_{l=1}^L \Tr \left( \bM_q {\bar{\bi}_l}^{\star} \left( {\bar{\bi}_l}^{\star} \right)^H \right) \tau_l &=   \sum \limits_{l=1}^L \Tr \left( \bM_q \bv_l \bv_l^H \right) \lambda_l \nonumber \\
  &= \Tr \left( \bM_q \bX^{\star}\right),
\end{align}
and the total transmit power is given by
\begin{align}
  \frac{1}{2} \sum \limits_{l=1}^L \Tr \left( {\overline{\bB}} {\bar{\bi}_l}^{\star} \left( {\bar{\bi}_l}^{\star} \right)^H \right) \tau_l  &=   \sum \limits_{l=1}^L \Tr \left( {\overline{\bB}} \bv_l \bv_l^H \right) \lambda_l \nonumber \\
  &= \frac{1}{2} \Tr \left( {\overline{\bB}} \bX^{\star}\right).
\end{align}
Clearly, by using the above TS scheme, the delivered power and the total transmit power are the same as those by using the solution $\bX^{\star}$ to $(\mathrm{P2}\mathrm{-SDR})$. Hence, the proof is completed.
\end{proof}

In general, since the optimal value of $(\mathrm{P2}\mathrm{-SDR})$ is a lower bound of that of $(\mathrm{P2})$, the above TS scheme thus achieves a TX sum-power that is no larger than the the optimal value of $(\mathrm{P2})$. Thus, the resulting solution can be considered to be optimal for (P2) if TS is allowed. Notice that in such cases, TS is required to achieve the boundary point of the multi-user power region with the given power profile vector $\bm{\alpha}$. \textcolor{black}{In summary, the aforementioned procedure to solve $\mathrm{(P2)}$ is given in Algorithm~\ref{AlgorithmP2}.}
\begin{algorithm}[t!] \small
\caption{: Algorithm for $\mathrm{(P2)}$ with TS} \label{AlgorithmP2}
\begin{algorithmic}[1]
\STATE Input parameters: ${\overline{\bB}}, w, P, \bM_q, r_{\textrm{rx}, q}, \alpha_q$, for $q=1,\ldots,\; Q.$ \\
		\STATE Solve $(\mathrm{P2}\mathrm{-SDR})$, obtain its solution as $\bX^{\star}$.
        \IF{$\rank{ \left( \bX^{\star} \right)}=1$} 
        \RETURN ${\bar{\bi}}^{\star} = \sqrt{\lambda_1} \bv_1$. (TS is not applied)
        \ELSE
        \RETURN ${\bar{\bi}_l}^{\star} = \sqrt{\sum \nolimits_{k=1}^L \lambda_k} \bv_l$, and $\tau_l = \frac{\lambda_l}{\sum \nolimits_{k=1}^L \lambda_k}$, for $l=1, 2,\ldots,\;L$. (TS is applied)
        \ENDIF
\end{algorithmic}
\end{algorithm}
\subsection{Solution to ($P1$) with All Constraints} \label{sec:solution_w_all_cons}
In this subsection, we consider ($\mathrm{P1}$) with all the constraints. Denote the space of $N$-order complex matrices by $\bbC^{N \times N}$. Let $\bX = \bi \bi^H$. The SDR of $\mathrm{(P1)}$ is given by
\begin{subequations}\label{eq:optimP1-SDR}
\begin{align}
(\mathrm{P1}\mathrm{-SDR}): \ \  \underset{ \bX \in \bbC^{N \times N}}{\text{min}} \ \ &\frac{1}{2}   \Tr \left( {\overline{\bB}}  \bX \right) \label{eq:rewardP1SDP} \\
\text{s.t.} \ \ &\Tr \left( {\bM_q} \bX \right) \geq \frac{2 r_{\textrm{rx}, q} \alpha_q P}{w^2},          \nonumber \\
         &\qquad \quad  q=1, 2, \ldots, Q  \label{eq:const1P1SDP} \\
&\Tr \left( {\bB_n} \bX \right) \leq V_{ n}^2, \nonumber \\
&\qquad \quad  n=1,\ldots, N \label{eq:const2P1SDP} \\
&\Tr \left( \bW_n \bX \right) \leq A_{ n}^2, \nonumber \\
&\qquad \quad  n=1,\ldots, N \label{eq:const3P1SDP} \\
&\ \bX \succcurlyeq 0.  \label{eq:const4P1SDP}
\end{align}
\end{subequations}

Like $(\mathrm{P2}\mathrm{-SDR})$, $(\mathrm{P1}\mathrm{-SDR})$ is also convex. By exploiting its structure, we obtain the following result on the rank of the optimal solution to $(\mathrm{P1}\mathrm{-SDR})$.
\begin{mythe}\label{the:rank-bound}
The rank of the optimal solution $\bX^{\star \star}$ to $(\mathrm{P1}\mathrm{-SDR})$ is upper-bounded by
\begin{align}
  \rank  \left( \bX^{\star  \star} \right) \leq \min \left(Q, \sqrt{Q+2N}\right). \label{eq:rank_bound}
\end{align}
\end{mythe}
\begin{proof}
Please refer to Appendix~\ref{sec:proof_theorem1}.
\end{proof}
The optimal solution $\bX^{\star  \star}$ to $(\mathrm{P1}\mathrm{-SDR})$ can be efficiently obtained by CVX~\cite{CVXTool}. Moreover, from Theorem~\ref{the:rank-bound}, we directly obtain the following corollary.
\begin{mycor}\label{the:rank-one_Q1}
  For $(\mathrm{P1})$ in the case of $Q=1$, the SDR in $(\mathrm{P1}\mathrm{-SDR})$ is tight, i.e., the optimal solution $\bX^{\star \star}$ to $(\mathrm{P1}\mathrm{-SDR})$ is always of rank-one with $\bX^{\star \star} = \bi^{\star \star } \left(\bi^{\star \star}\right)^H$, where $\bi^{\star \star}$ is thus the optimal solution to $(\mathrm{P1})$.
\end{mycor}
For the general case of $Q>1$, if the solution $\bX^{\star \star}$ to $(\mathrm{P1}\mathrm{-SDR})$ is of rank-one with $\bX^{\star \star} = \bi^{\star \star} \left(\bi^{\star \star}\right)^H$, then $\bi^{\star \star}$ is the optimal solution to $(\mathrm{P1})$; however, for the case of $\rank\left( \bX^{\star \star}\right) > 1$, in the following we propose  two approximate solutions for $(\mathrm{P1})$ based on TS and randomization, respectively.
\subsubsection{TS-based Solution}
We note that the TS scheme proposed in Section~\ref{sec:MU_woPeak} for $(\mathrm{P2})$ cannot be directly applied to $(\mathrm{P1})$ due to the additional peak voltage and current constraints. This is because the current solutions given in \eqref{eq:current_TS_woPeak} in general may not satisfy these peak constraints at all TXs over all the $L$ time slots. To tackle this problem, we treat the time allocation $\tau_l$'s and the current scaling factors, denoted by $\sqrt{\theta_l}$ with $\theta_l \geq 0,\ \forall l=1,\ldots,L$, for all slots as design variables, such that all peak constraints can be satisfied over all slots. Recall $\tau_l$'s are subject to $\sum_{t=1}^L \tau_l=1$, and $\tau_l \geq 0, \ \forall l$; and with a little abuse of notations, we still use $\bv_l$'s to denote the singular vectors obtained from the SVD of the optimal solution $\bX^{\star \star}$ to $(\mathrm{P1}\mathrm{-SDR})$, similar to those defined for $\bX^{\star \star}$ to $(\mathrm{P2}\mathrm{-SDR})$. In the $l$-th slot, the TX current vector is then set as
\begin{align}
{\bar{\bi}_l} &= \sqrt{\theta_l} \bv_l. \label{eq:current_TS_wPeak}
\end{align}

Let $\bm{\theta}=[\theta_1 \; \ldots \ \theta_L]^T$, and $\bm{\tau}=[\tau_1 \; \ldots \ \tau_L]^T$. Moreover, we denote $\bV_l=\bv_l \bv_l^H$, and nonnegative constants $c_{0, l}= \Tr \left( {\overline{\bB}}  \bV_l \right)$, $c_{1, lq} = \Tr \left( {\bM_q} \bV_l \right)$, $c_{2, ln} = \Tr \left( {\bB_n} \bV_l \right)$, and $c_{3, ln} = \Tr \left( \bW_n \bV_l \right)$. We then formulate the following problem to obtain the TS-based solution for $(\mathrm{P1})$.
\begin{subequations}\label{eq:optimP1TS}
\begin{align}
(\mathrm{P1}\mathrm{-TS}): \ \ \underset{\bm{\theta}, \; \bm{\tau}}{\text{min}} \ \
&\sum \limits_{l=1}^L \frac{c_{0, l} \theta_l \tau_l}{2} \label{eq:rewardP1TS} \\
\text{s.t.} \ \
&\sum \limits_{l=1}^L c_{1, lq} \theta_l \tau_l \geq \frac{2 r_{\textrm{rx}, q} \alpha_q P}{w^2},  \nonumber \\
&\qquad  q=1,\ldots, Q \label{eq:const1P1TS} \\
&c_{2, ln} \theta_l \leq V_{ n}^2,  \nonumber \\
& \quad   n=1,\ldots, N, \;~ l=1,\ldots, L \label{eq:const2P1TS} \\
&c_{3, ln} \theta_l \leq A_{ n}^2,  \nonumber \\
& \quad  n=1,\ldots, N, \;~ l=1,\ldots, L  \\
&\sum_{t=1}^L \tau_l=1, \label{eq:const4P1TS} \\
&\tau_l \geq 0, \; \theta_l \geq 0, \;~ l=1,\ldots, L. \label{eq:const3P1TS}
\end{align}
\end{subequations}

We define a set of new variables as $\phi_l=\theta_l \tau_l, \ l=1,\ldots, L$. Problem $(\mathrm{P1}\mathrm{-TS})$ is thus rewritten as the following linear-programming (LP), which can be efficiently solved by e.g., CVX~\cite{CVXTool}.
\begin{subequations}\label{eq:optimP1TSLP}
\begin{align}
(\mathrm{P1}\mathrm{-TS-LP}): \ \ &\underset{\bm{\lambda}, \; \bm{\tau}}{\text{min}} \ \
\sum \limits_{l=1}^L \frac{c_{0,  l} \phi_l}{2} \label{eq:rewardP1TSLP} \\
\quad \text{s.t.} \ \
&\sum \limits_{l=1}^L c_{1, lq} \phi_l \geq \frac{2 r_{\textrm{rx}, q} \alpha_q P}{w^2}, \nonumber \\
&\quad    q=1,\ldots,  Q \label{eq:const1P1TSLP} \\
&c_{2, ln} \phi_l - V_{ n}^2 \tau_l \leq 0, \nonumber \\
&\quad    n=1,\ldots, N, \;~ l=1, \ldots, L \label{eq:const2P1TSLP} \\
&c_{3, ln} \phi_l - A_{ n}^2 \tau_l \leq 0, \nonumber \\
&\quad    n=1,\ldots, N, \;~ l=1, \ldots, L \label{eq:const3P1TSLP} \\
&\sum_{t=1}^L \tau_l=1, \label{eq:const4P1TSLP} \\
&\tau_l \geq 0, \; \phi_l \geq 0, \;~ l=1,\ldots, L. \label{eq:const3P1TSLP}
\end{align}
\end{subequations}
If the above $(\mathrm{P1}\mathrm{-TS-LP})$ is feasible, there is a feasible TS-based solution for $(\mathrm{P1})$; otherwise $(\mathrm{P1})$ is regarded as infeasible, which implies that the RX sum-power $P$ needs to be decreased in the next bisection search iteration in Algorithm~\ref{Algorithm}.
\subsubsection{Randomization-based Solution}
The randomization technique is a well-known method applied to extract a feasible approximate QCQP solution from its SDR solution. Before presenting the proposed randomization-based solution, we first describe the steps for generating feasible random vectors from SDR solution. Recall the SVD of $\bX^{\star \star}$ as $\bX^{\star \star}=\bV \bLambda \bV^H$. Define $\bLambda^{\frac{1}{2}} \triangleq \diag \{\sqrt{\lambda_1}, \ldots,  \sqrt{\lambda_L}\}$. A random vector is \textcolor{black}{specifically} generated as follows:
\begin{align}
\by_d =\bV \bLambda^{\frac{1}{2}} \bw_d,
\end{align}
where $\bw_d \sim \calC \calN(\bold{0}_N, \bI_N)$, with $\bold{0}_N$ representing an all-zero column vector of length $N$.

To further generate a random vector $\bx_d$ that is feasible to $\mathrm{(P1)}$, we scale the vector $\by_d$ by $\mu_d$ with $\mu_d \in \bbR$, i.e., $\bx_d \triangleq \mu_d \by_d$. If the resulting problem shown as follows is feasible, a feasible $\mu_d$ is thus found; otherwise no feasible vector can be obtained from this $\by_d$.
\begin{subequations}\label{eq:optimP1_rand_feasibility}
\begin{align}
\mathrm{find}: \ \ &\mu_d  \label{eq:fea} \\
\text{s.t.} \ \
&\frac{w^2 \mu_d^2}{2 r_q}  {\by_d}^H \bM_q {\by_d} \geq \alpha_q P, \;~ q=1,\ldots, Q \label{eq:const1Pfea} \\
&\mu_d^2 \by_d^H \bB_n \by_d \leq V_{ n}^2, \;~   n=1,\ldots, N \label{eq:const2Pfea} \\
&\mu_d^2 \by_d^H \bW_n \by_d \leq A_{ n}^2, \;~ n=1,\ldots, N. \label{eq:const3Pfea}
\end{align}
\end{subequations}
The proposed algorithm for obtaining the randomization-based solution \textcolor{black}{is summarized} as Algorithm \ref{Algorithm_rand}.

\section{Magnetic Channel Estimation}\label{sec: CE}
For implementation of magnetic beamforming in practice, it is necessary for the central controller at the TX side to estimate the mutual inductance between each pair of TX coil and RX coil, namely magnetic MIMO channel estimation.
\textcolor{black}{Note that in this paper,  the mutual inductances $M_{nq}$'s are assumed to be quasi-static, i.e., they remain  constant over a certain block of time, but may change from one block to another, since the RXs are mobile devices in general. Hence, $M_{nq}$'s need to be estimated periodically over time.
For practical implementation, at the beginning of each transmission period, we treat all the magnetic channels $M_{nq}$'s as unknown real parameters.  For convenience, we denote the magnetic channel matrix by $\bM$ with elements given by $M_{nq}$'s.
In the next, we first consider magnetic MIMO channel estimation for the ideal case with perfect RX current knowledge and then the  practical case with imperfect current  knowledge.}

\begin{algorithm}[t!]
\caption{: Randomization-based Solution for $\mathrm{(P1)}$} \label{Algorithm_rand}
\begin{algorithmic}[1]
\STATE Initialization: the solution $\bX^{\star \star}$ to $(\mathrm{P1}\mathrm{-SDR})$, a large positive integer $D$ \textcolor{black}{(set as $D=4\times 10^{3}$ in our simulations)}, set $\calD = \emptyset$. \\
\STATE Compute the SVD of $\bX^{\star \star}$ as $\bX^{\star \star}=\bV \bLambda \bV^H$. \\
\FOR{$d=1,\ldots,D$}
\STATE Generate a random vector $\by_d =\bV \bLambda^{\frac{1}{2}} \bw_d$, where $\bw_d \sim \calC \calN(\bold{0}_N, \bI_N)$.
\IF{the problem~\eqref{eq:optimP1_rand_feasibility} is feasible, }
\STATE Obtain $\bx_d =\mu_d \by_d$.
\STATE $\calD = \calD \bigcup d$.
\ENDIF
\ENDFOR
\RETURN $\bi^{\star \star} = \arg \underset{d \in \calD} {\min} \quad \frac{1} {2} { {{\bx_d}}^H  {\overline{\bB}} {{\bx_d}}}$ if $\calD \neq \emptyset$; otherwise, declare $(\mathrm{P1})$ is infeasible.
\end{algorithmic}
\end{algorithm}

We assume that each RX $q$ can feed back its measured current to the central controller by using existing communication module. One straightforward method to estimate $M_{nq}$ is given in~\cite{RezaZhang16TISPN}, where by switching off all the other TXs and RXs, TX $n$ can estimate $M_{nq}$ with RX $q$ based on the current measured and fed back by RX $q$. However, this method may not be efficient for estimating the magnetic MIMO channel $\bM$, since it requires synchronized on/off operations of all TXs and RXs and also needs at least $NQ$ iterations to estimate all $M_{nq}$'s.
Alternatively, we propose more efficient methods that can simultaneously estimate the magnetic MIMO channel $\bM$ in $T$ ($T \geq Q$) time slots. In the $t$-th slot, we apply a source voltage $v_{\textrm{tx}, n, t}$ on TX $n$, and the current $i_{\textrm{tx}, n, t}$ is measured by TX $n$. From Kirchhoff's circuit laws, the voltage of TX $n$ is
\begin{align}
  v_{\textrm{tx},n, t}&=r_{\textrm{tx}, n} i_{\textrm{tx}, n, t} + \nonumber \\
  &\quad j w \sum \limits_{k=1, \neq n}^N \tilM_{nk} i_{\textrm{tx}, k, t} - j w  \sum \limits_{q=1}^Q M_{nq} i_{\textrm{rx}, q, t}. \label{eq:voltage_CE}
\end{align}

In practice, randomly generated voltage values are assigned over different TXs as well as over different time slots.

Define the $N \times T$ matrices $\bH$ and $\bY$ with elements given by $v_{\textrm{tx}, n, t}$'s and $i_{\textrm{tx}, n, t}$'s, respectively. Moreover, define the $Q \times T$ matrix $\bZ$ with elements given by $i_{\textrm{rx}, q, t}$'s and the $N \times N$ matrix $\bF$ with elements given by
\begin{align}
 F_{nk}  &=
  \left\{ \begin{array}{cl}
  r_{\textrm{tx}, n} , &\mbox{if}\; k=n\\
    j w \tilM_{nk}, \qquad &\mbox{otherwise}. \\
  \end{array}
  \right. \label{eq_F}
\end{align}
Since the fixed TX-TX mutual inductance $\tilM_{nk}$ can be measured offline and the TX currents $i_{\textrm{tx}, n, t}$'s as well as voltages $v_{\textrm{tx}, n, t}$'s can be measured by the TXs, the matrices $\bF$ and $\bY$ are assumed to be known by the central controller perfectly.
From~\eqref{eq:voltage_CE}, the voltages at all TXs over $T$ time slots can be written in the following matrix-form
\begin{align}\label{eq:CEeq0}
  \bH = \bF \bY - j w \bM \bZ.
\end{align}
Let $\bG \triangleq \frac{j}{w} (\bH - \bF \bY)$. The voltage matrix in  \eqref{eq:CEeq0} can be rewritten as
\begin{align}\label{eq:CEeq1}
  \bG = \bM \bZ.
\end{align}
With known $\bH, \ \bF$ and $\bY$, the matrix $\bG$ is known by the central controller.
\subsection{Channel Estimation with Perfect RX-Current Knowledge} \label{sec:MCE_perfectCurrent}
For the case with perfect RX-current knowledge of $\bZ$ at the central controller, it suffices to use $Q$ time slots for channel estimation, i.e., $T=Q$. Since the voltage values are randomly generated and assigned over different TXs as well as over different time slots, the RX current matrix $\bZ$ known at the central controller can be assumed to have a full rank of $Q$ and thus its inverse exists. Hence, the mutual inductance matrix $\bM$ can be estimated as
\begin{align}\label{eq:Mestimate_perfectmeasure}
  \hatbM = \bG \bZ^{-1}. 
\end{align}
Note that from~\eqref{eq_F} and~\eqref{eq:CEeq0}, it can be shown that the estimate in~\eqref{eq:Mestimate_perfectmeasure} is always a real matrix.





\subsection{Channel Estimation with Imperfect RX Current Knowledge} \label{sec:MCE_imperfectCurrent}


In practice, the RX-current information of $\bZ$ obtained by the central controller are not perfect, due to various \textcolor{black}{errors} such as the error in the current meter reading, quantization error and feedback error, etc. Denote the error of the $q$-th RX's current in the $t$-th slot by $e_{\textrm{rx}, q, t}$. We assume that all the current errors $e_{\textrm{rx}, q, t}$'s are mutually independent and each follows the CSCG distribution with zero mean and variance $\sigma^2$. The corresponding RX current known by the central controller is thus $i_{\textrm{rx}, q, t}^{\prime} = i_{\textrm{rx}, q, t} + e_{\textrm{rx}, q, t}$. Denote all the RX-current errors by the $Q \times T$ matrix $\bE$ with elements $e_{\textrm{rx}, q, t}$'s. In addition, denote the RX-current knowledge obtained at the central controller by the $Q \times T$ matrix $\tilbZ$ with elements $i_{\textrm{rx}, q, t}^{\prime}$'s. We thus have
\begin{align}\label{eq:error_model}
  \bZ &= \tilbZ - \bE.
\end{align}
With RX-current errors, from~\eqref{eq:error_model}, the circuit equation in~\eqref{eq:CEeq1} is rewritten as follows:
\begin{align}\label{eq:CEeq}
  \bG = \bM \tilbZ - \bM \bE.
\end{align}


\textcolor{black}{In the following, we first show the difficulty to obtain the maximum likelihood (ML) estimate for the magnetic channel $\bM$, then present a suboptimal but efficiently implementable least-square (LS) \textcolor{black}{based} estimate for $\bM$.}
Define $\bA=\bM \bE$ for convenience. From~\eqref{eq:CEeq}, we have
\begin{align}\label{eq:CEeq2}
   \bA = \bM \tilbZ - \bG.
\end{align}
Denote the columns of $\bA, \ \tilbZ$ and $\bG$ by $\ba_t, \ \tilbz_t$ and $\bg_t$, respectively, for $t=1,\ldots, T$. Then $\ba_t$ is a CSCG random vector with mean $\bm{\mu}_t(\bM)$ and covariance matrix $\bSigma (\bM)$, \textcolor{black}{which are} given by
\begin{align}
  \bm{\mu}_t(\bM) &=  \bM \tilbz_t - \bg_t, \\
  \bSigma (\bM) &= \sigma^2 \bM \bM^T.
\end{align}
From the mutual independence of $\ba_t$'s, the joint probability \textcolor{black}{distribution} of $\ba_t$'s \textcolor{black}{is} given by
\begin{align}
  &p(\bA) = \frac{1}{(2 \pi)^{\frac{NT}{2}} \left|\bSigma (\bM)\right|^{\frac{T}{2}}}  \\
  &\quad \exp \left( -\frac{1}{2} \sum \limits_{t=1}^T (\ba_t - \bm{\mu}_t (\bM))^H (\bSigma (\bM))^{-1} (\ba_t - \bm{\mu}_t (\bM))\right).\nonumber
\end{align}
The log-likelihood function of the above probability density function (PDF) is thus
\begin{align}\label{eq:log_likelihood}
  \log(\calL) &= - T \log (|\bSigma (\bM)|) - N T \log(\pi) - \nonumber\\
  &\sum \limits_{t=1}^T (\ba_t - \bm{\mu}_t (\bM))^H (\bSigma (\bM))^{-1} (\ba_t - \bm{\mu}_t (\bM)).
\end{align}
The ML estimate should be obtained by maximizing the log-likelihood function $\log(\calL)$ in~\eqref{eq:log_likelihood} over $\bM$. To this end, we take the derivative of $\log(\calL)$ with respect to $\bM$ as follows:
\begin{align}\label{eq:der_M}
  &\frac{\partial \log(\calL)}{\partial \bM} =
  - T  \frac{\partial \log(|\bM \bM^T|)}{\partial \bM}-  \\
  &\sum \limits_{t=1}^T \frac{\partial}{\partial \bM} \left[(\ba_t - \bm{\mu}_t (\bM))^H \left( \sigma^2 \bM \bM^T \right)^{-1} (\ba_t - \bm{\mu}_t (\bM)) \right].\nonumber
\end{align}
\textcolor{black}{However, it is difficult to simplify the derivative in~\eqref{eq:der_M} to derive the optimal $\bM$, since the means $\bm{\mu}_t(\bM)$'s depend on the unknown $\bM$ and also vary over $t$, and furthermore the covariance matrix $\bSigma (\bM)$ is a scaled Gramian matrix of $\bM^T$. Hence, we present a suboptimal LS estimate, denoted by $\hatbM_{\sf LS}$, for $\bM$ in the following theorem.}

\begin{mythe}\label{the:LS}
The LS estimate of $\bM$ is given by
\begin{align}\label{eq:LS_est}
  \hatbM_{\textrm{LS}} = \left( \bG \tilbZ^H  + \bG^{\ast} \tilbZ^T \right) \left( \tilbZ \tilbZ^H  + \tilbZ^{\ast} \tilbZ^T \right)^{-1}, 
\end{align}
and the resulting squared error is given by
\begin{align}\label{eq:LSE}
  J = \Tr \left( \left( \bG - \hatbM_{\textrm{LS}} \tilbZ \right) \left( \bG - \hatbM_{\textrm{LS}} \tilbZ \right)^H\right).
\end{align}
\end{mythe} 
\begin{proof}
  Please refer to Appendix~\ref{sec:proof_MCE}.
\end{proof}


It is noted that the LS estimate $\hatbM_{\textrm{LS}}$ in~\eqref{eq:LS_est} is always a real matrix, although both $\bG$ and $\tilbZ$ are complex matrices in general.

 \begin{table*} [t!]
	\centering
	\caption{Mutual/Self inductance values ($\mu$H)} \label{table_inductance_new}
	\footnotesize{
		\textcolor{black}{\begin{tabular}{*{10}{c}}
			\hline \hline
			& TX $1$ & TX $2$ & TX $3$ & TX $4$ & TX $5$ &RX $1$ &RX $2$ &RX $3$ &RX $4$\\
			\hline
			TX $1$   & $47700$  & $2.2970$ & $0.8074$ & $2.2970$  & $6.5741$  &$0.9468$ &$0.04747$ &$0.02789$ &$0.03874$ \\
			TX $2$  & $2.2970$  & $47700$ & $2.2970$  & $0.8074$  & $6.5741$  &$0.01733$ &$0.5642$ &$0.05711$ &$0.01733$ \\
			TX $3$   & $0.8074$  & $2.2970$ & $47700$  & $2.2970$  &$6.5741$ &$0.007872$ &$0.01945$ &$0.09880$ &$0.03874$ \\
			TX $4$   & $2.2970$  & $0.8074$ & $2.2970$  & $47700$  & $6.5741$ &$0.02817$ &$0.01116$ &$0.03825$ &$0.2458$ \\
			TX $5$   & $6.5741$  & $6.5741$ & $6.5741$ & $6.5741$  & $47700$  &$0.07472$ & $0.1526$ &$1.3266$ &$0.5256$ \\
			RX $1$   & $0.9468$ & $0.01733$ &  $0.007872$ & $0.02817$  & $0.07472$  & $280.32$ &$0.0003932$ &$0.0003153$ & $0.0005579$ \\
			RX $2$   & $0.04747$  & $0.5642$ & $0.01945$ & $0.01116$ &$0.1526$ &$0.0003932$ &$280.32$ &$0.001130$ & $0.0003153$ \\
			RX $3$   &  $0.02789$ & $0.05711$ & $0.09880$ & $0.03825$  & $1.3266$  &$0.0003153$ &$0.001130$ &$280.32$ & $0.002561$ \\
			RX $4$   & $0.03874$ & $0.01733$ &$0.03874$ & $0.2458$ & $0.5256$  &$0.0005579$ &$0.0003153$ &$0.002561$ & $280.32$ \\
			\hline
		\end{tabular}
	}}
\end{table*}


\section{Numerical Results}\label{sec: simulation}
\textcolor{black}{In this section, we evaluate the performance of our proposed magnetic channel estimation and magnetic beamforming schemes.
As shown in Fig.~\ref{fig:Fig_app}, we consider a MIMO MRC-WPC system, which constitutes a rectangular table of size $1.6\text{m}\times 1.6\text{m}$ with $N=5$ built-in wireless chargers placed horizontally below its surface and $Q=4$ RXs  placed horizontally on its surface  at \textcolor{black}{random} locations.}
\textcolor{black}{Specifically,  we consider the thickness of the table's surface is $10$cm, which is indeed the same as the vertical separating distance between each TX and RX.
For TXs $1$--$5$, we set $(x=0.7,y=0.7)$, $(x=-0.7,y=0.7)$, $(x=-0.7,y=-0.7)$, $(x=0.7,y=-0.7)$, and $(x=0,y=0)$, respectively, in meter. On the other hand, for RXs $1$--$4$, we set $(x=0.7,y=0.5)$, $(x=-0.3,y=0.6)$, \textcolor{black}{$(x=-0.2,y=-0.1)$}, and $(x=0.3,y=-0.3)$, respectively, in meter.}
\textcolor{black}{Moreover, we consider that each TX coil has $250$ turns and a radius of $10$cm, while each RX coil has $50$ turns and a radius of $2$cm.
We assume that coils are all made from  copper wire with radius of $0.25$mm.
We set the resistance of each TX $n$ as $r_{\textrm{tx}, n}=13.44 \ \Omega, n=1,...,N$, which is set equal  to the ohmic resistance of its coil.}
\textcolor{black}{Similarly, the parasitic resistance of each \textcolor{black}{RX} $q$ is set as $r_{ \textrm{rx,p}, q}=0.5367~\Omega$. We also set the load resistance at \textcolor{black}{RX} $q$ as  $r_{ \textrm{rx,l}, q}=10~\Omega$.  Clearly, the parasitic resistance  of each RX  is negligible compared to the much larger load resistance, which is typical in practice. Hence, in our  simulations, we can safely set $r_{ \textrm{rx,p}, l}/r_{ \textrm{rx},q}\approx 1$, $q=1,\ldots,Q$.}
\textcolor{black}{The compensators' capacitances at the TXs and RXs are  set such that their natural angular frequencies all become identically   $w=42.6 \times 10^6$ rad/second (or  $6.78$ MHz).
Last, we set $P_T =100$~W, and the peak voltage/current constraints  at all TXs are given by $V_n=50\sqrt{2}$~V and $A_{n}=5\sqrt{2}$~A, $n=1,...,N$, respectively.}

\textcolor{black}{The self and mutual inductance values of all TXs and RXs are given in Table \ref{table_inductance_new}. From the mutual inductance values in Table~\ref{table_inductance_new}, we have two observations: first, on average the coupling between RXs are considerably smaller than that between RXs and TXs; second, for each RX the ratio of the mutual inductance between itself and the closest TX  to  that between itself and the closest RX is very large  (e.g., the ratios are $1700$, $500$, $520$, and $205$ for RXs $1$--$4$, respectively).
In this case, the current induced at each RX is mainly due to the magnetic flux generated by its nearby \textcolor{black}{TX(s)}, which is in accordance with our previous  assumption that the mutual inductances between RXs are negligible.
In the following simulations,  we thus ignore the mutual inductance between RXs.}

\subsection{Magnetic MIMO Channel Estimation}\label{sec:MSE_CE}
In this subsection, we evaluate the performance of magnetic channel estimation. For performance comparison, we use the channel estimation scheme in~\cite{RezaZhang16TISPN} as a benchmark, where in each training slot, only one pair of TX $n$ and RX $q$ are switched on, and the mutual inductance between this pair of coils is estimated as
\begin{align}\label{eq:benchmark}
\hatM_{nq} = \textrm{Re} \left\{\frac{r_{\textrm{tx}, n} i_{\textrm{tx}, n}-v_{\textrm{tx}, n}}{jw \tili_{\textrm{rx}, q}} \right\},
\end{align}
with the imperfect RX-current $\tili_{\textrm{rx}, q}=i_{\textrm{rx}, q}+e_{\textrm{rx}, q}$. The real-part operation is adopted, since the RX-current error $e_{\textrm{rx}, q}$ is complex in general. The total required number of slots is thus $NQ=20$ in our numerical example.
For the proposed LS estimation scheme, we assume that the training period of $T$ time slots can be divided into multiple blocks each of which consists of $N$ successive time slots. In the $n$-th time slot of each block, TX $n$ carries a source voltage of $0.75$ V, and other TXs remain in closed-loop but with a source voltage of $0$ V, which ensures that the total power consumed by all TXs in each slot is $40$~W \textcolor{black}{(i.e., less than  $P_T=100$~W)}. We define the signal-to-noise ratio (SNR) of the RX-current estimation as $\gamma = \frac{\bbE \left[ |i_{\textrm{rx}, q, t}|^2 \right]}{\sigma^2}$, where the expectation is with respect to different $q$'s and $t$'s. We define the following normalized mean squared error (MSE) as the performance metric,
\begin{align}
\bar{\varepsilon} &=   \frac{ \bbE_{\hatbM} \left[ \| \bM - \hatbM\|_F^2\right] }{\| \bM \|_F^2}.
\end{align}
The following numerical results are based on $10^6$ Monte Carlo simulations each with randomly generated RX current errors.

Fig.~\ref{fig:Fig0} plots the normalized MSE $\bar{\varepsilon}$ versus the RX-current SNR $\gamma$, for both our proposed LS estimation scheme and the benchmark scheme in~\cite{RezaZhang16TISPN}.
For \textcolor{black}{the proposed} scheme, we observe that in general $\bar{\varepsilon}$ decreases as $\gamma$ increases, as expected. In particular, for $T=10$, we observe that the normalized MSE is $2.8 \times 10^{-3}$, $3 \times 10^{-4}$, and $3 \times 10^{-5}$ for $\gamma=20$ dB, $30$ dB, and $40$ dB, respectively. In practice, the precision of current meters is typically more than $99\%$. Neglecting the quantization error and feedback error, the SNR $\gamma$ can thus be practically modeled as $40$ dB. Since the proposed channel estimation scheme is accurate enough in practice, we thus assume that the RX-current information is known by the central controller perfectly in the subsequent simulations.
Moreover, we observe that the MSE decreases as the number of training slots $T$ increases.
On the other hand, we observe that the LS estimation outperforms the estimation in~\cite{RezaZhang16TISPN} in terms of \textcolor{black}{both} lower MSE and less training time \textcolor{black}{required}. For an MSE level of $10^{-3}$, the LS estimation achieves a SNR improvement of 3 dB and 6 dB for $T=10$ and $20$, respectively. Also, for the scheme in~\cite{RezaZhang16TISPN}, the MSE is high in the low SNR region, due to the real-part rounding operation in~\eqref{eq:benchmark}.
\begin{figure}[t!]
	\centering	\includegraphics[width=.95\columnwidth]{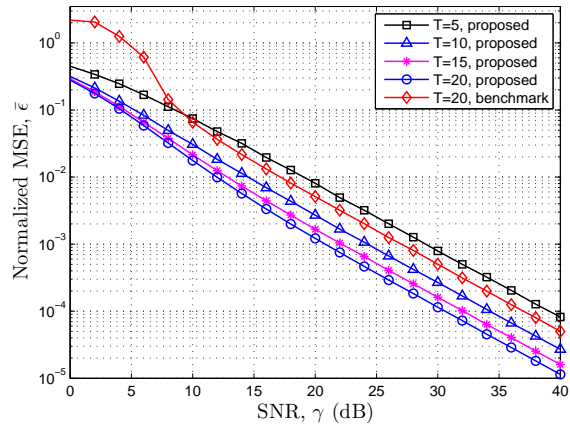}
	\caption{Normalized MSE for TX-RX inductance estimation.}
	\label{fig:Fig0}
\end{figure}

\begin{table*} [t!]
	\centering
	\caption{Comparison under different load power} \label{table_solution}
	\scriptsize{
		\begin{tabular}{*{3}{c}}
			\hline \hline
			& $P=1$~W & $P=56$~W\\
			\hline
			$(i_1^{\star}, v_1^{\star}, p_1^{\star})$    & $(-0.0152,  \ -1.109 -32.027j, \ 0.0085)$ & $(-0.224, \ -52.910 +46.910j,\ 5.9279)$ \\
			$(i_2^{\star}, v_2^{\star}, p_2^{\star})$   & $(-0.181,  \ -13.185 -15.953j, \ 1.194)$ & $(1.269 + 0.786j, \ 68.983 -15.531j, \ 37.661)$ \\
			$(i_3^{\star}, v_3^{\star}, p_3^{\star})$   & $(-0.0062,  \ -0.454 -32.336j, \ 0.0014)$ & $(-0.190 + 0.0036j,  \ -55.667 +43.602j, \ 5.381)$ \\
			$(i_4^{\star}, v_4^{\star}, p_4^{\star})$   & $(-0.0036,  \ -0.260 -22.0638j, \ 0.000467)$ & $(-0.702 - 0.573j, \ -70.073 - 9.468j, \ 27.321)$\\
			$(i_5^{\star}, v_5^{\star}, p_5^{\star})$   & $(-0.0490,  \ -3.565 -57.779j,\ 0.0874)$ & $(-0.0204 + 0.123j, \ -42.861 +56.239j, \ 3.906)$\\
			\hline
		\end{tabular}
	}
\end{table*}

\subsection{MISO WPT with Single RX }\label{sec:simulation_MOSO}
In this subsection, we consider the special case of a single RX,  i.e., only RX $2$ is present in Fig. \ref{fig:Fig_app}.
For performance benchmark, we consider an uncoordinated WPT system with all TXs set to have identical current with equal power consumption. We compare this system with our proposed coordinated WPT with optimal magnetic beamforming without or with the peak voltage and current constraints at all TXs.
\textcolor{black}{We define the efficiency of WPT as the ratio of the delivered load power $P$ to the total TX power $p_{\textrm{tx}}$, i.e., $\eta \triangleq \frac{P}{p_{\textrm{tx}}}$.}

Fig.~\ref{fig:Fig_A_fixedR} plots the total TX power $p_{\textrm{tx}}$ and the efficiency $\eta$ versus the delivered load power $P$. For the case without TX voltage/current constraints, it is observed that the WPT efficiencies with magnetic beamforming and benchmark system are $77.3\%$ and $58.6\%$, respectively.
For the case with TX voltage/current constraints, it is observed that magnetic beamforming can deliver power up to $56$~W to the RX with the efficiency of $70\%$; while the benchmark system can deliver at most $0.2$~W to the RX with the efficiency of $58.6\%$. Thus, besides the WPT efficiency improvement, magnetic beamforming also significantly enhances the maximum power deliverable to the RX load, under the same practical circuit constraints.

Fig.~\ref{fig:Fig_A_fixedR} also shows that the WPT efficiency decreases over $1\le P \le 56$ in W.
To explain this observation and obtain insights for magnetic beamforming, we further investigate the two cases of $P=1$~W and $56$~W in the following. The optimal currents, the corresponding voltages and the consumed powers of all TXs are given in Table~\ref{table_solution} for these two cases. For $P=1$~W, it is observed that most of the transmit power is consumed by TX $2$ and TX $5$ which have the two largest mutual inductance values with the RX.
This implies that the TX with larger mutual inductance with the RX carries higher current, and thus consumes more power so as to maximize the efficiency of WPT.
In this case, all TX current or voltage constraints are inactive, and it can be further verified that the current of each TX is exactly proportional to its mutual inductance with the RX. This is in accordance with Theorem \ref{the:solution_specialP2_Q1}.
In contrast, to support higher RX load power of $56$~W, the voltages of all TXs reach the peak value $50 \sqrt{2}$ V. This results in a decreased efficiency, due to relatively smaller mutual inductance between TXs $1$, $3$, $4$ and the RX, compared to those between TXs $2$, $5$ and the RX.

\begin{figure} [t!]
	\centering
	\includegraphics[width=.95\columnwidth]{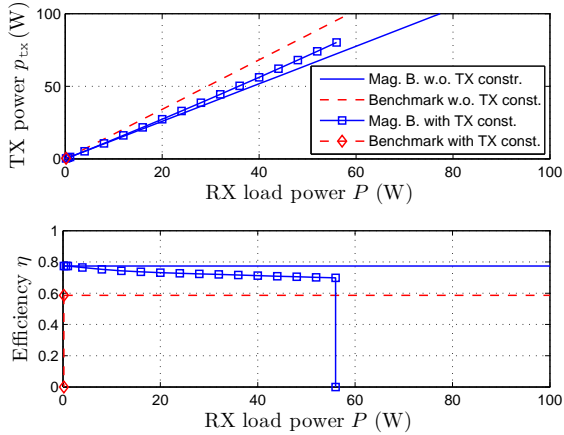}
	\caption{TX sum-power and efficiency versus RX load power.}
	\label{fig:Fig_A_fixedR}
\end{figure}

\subsection{MIMO WPT with Multiple RXs}
In this subsection, we consider the multi-user case, i.e., there are \textcolor{black}{more than one} RXs.
\subsubsection{Two-user Case}
For the two-user case, we consider in Fig. \ref{fig:Fig_app} only RXs 1 and 2 are present. Fig.~\ref{fig:Fig_region} plots the power regions for the proposed magnetic beamforming and the benchmark \textcolor{black}{scheme with uncoordinated WPT}, respectively. Each power region is shown as a convex set, as expected. There is a \textcolor{black}{trade-off} between the maximally delivered powers $p_{\textrm{rx}, \; 1}$ and $p_{\textrm{rx}, \; 2}$ to RX $1$ and RX $2$, respectively, i.e., $p_{\textrm{rx}, \; 2}$ decreases as $p_{\textrm{rx}, \; 1}$ increases. Without TX peak voltage/curent constraints, we observe that for the magnetic beamforming system, the maximally delivered power is $87.5$~W and $77.5$~W, for RX $1$ and RX $2$, respectively; while for the benchmark system, the maximally delivered power for them are $50.4$~W and $27.5$~W, respectively.

With TX peak voltage/current constraints, the maximally delivered power is $46$~W and $57.5$~W, for RX $1$ and RX $2$, respectively, for the magnetic beamforming system. This is because the inductance (i.e., $0.9468$ \textcolor{black}{$\mu$H}) between TX $1$ and RX $1$ is much larger than that between \textcolor{black}{any} other TX and RX $1$, and the maximally delivered power to RX $1$ is limited by the peak voltage constraint for TX $1$. In contrast, for the benchmark system, the maximally delivered power for them are $0.38$~W and $0.22$~W, for RX $1$ and RX $2$, respectively, which are negligible compared to those for the magnetic beamforming system.
The significant improvement over the benchmark system is also shown for both the cases with or without the TX peak constraints.

\begin{figure} [t!]
	\centering	\includegraphics[width=.95\columnwidth]{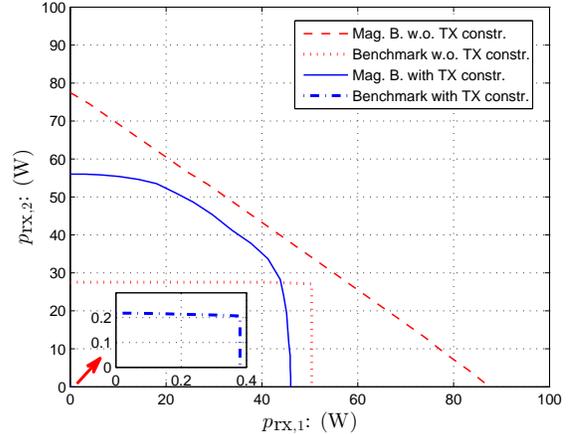}
	\caption{Power region for the two-user case with RX $1$ and RX $2$.}
	\label{fig:Fig_region} \vspace{-3mm}
\end{figure}

\subsubsection{Four-user Case}
For the four-user case, we consider RXs 1, 2, 3 and 4 are present in Fig.~\ref{fig:Fig_app}.
We fix the power profile vector $\bm{\alpha}=[\alpha_1 \ \alpha_2 \ \alpha_3 \ \alpha_4]^T =[0.1227 \ 0.03615 \ 0.7836 \ 0.05752]^T$, under which numerical results show that the optimal SDR solution is of rank two.
Fig.~\ref{fig:Fig_opt_subopt} plots the total TX power consumed versus the total RX power delivered for the TS-based solution, the randomization-based solution, and the benchmark scheme. We observe that the TS solution achieves the best performance, while the randomization solution performs slightly worse. For the benchmark scheme, the maximally delivered power to all RXs is \textcolor{black}{only} $0.8$~W, and the consumed total TX power increases faster than the proposed \textcolor{black}{schemes}.



\section{Conclusions}\label{sec:conslusion}
This paper has studied the optimal magnetic beamforming design subject to practical power and circuit constraints for the multi-user MIMO MRC-WPT system.
\textcolor{black}{To characterize the optimal performance trade-offs among the users on the boundary of the multi-user power region, we formulate an optimization problem to maximize the sum-power deliverable to all \textcolor{black}{RXs} subject to the constraints on the minimum load power at each \textcolor{black}{RX}, which is proportionally set based on a given power-profile vector, as well as the practical maximum peak voltage and current at each TX.
We propose an iterative algorithm to solve the formulated problem, which requires to solve a TX sum-power minimization problem at each iteration.}
For the special case of one single RX and without TX peak current/voltage constraints, the optimal current of each TX is shown to be proportional to the mutual inductance between its TX coil and the RX coil. Besides, for the case of multiple RXs and without TX peak current/voltage constraints, we propose a new TS-based scheme that achieves the optimal solution. In general, the TX sum-power minimization problem is a non-convex quadratically constrained quadratic programming (QCQP) and thus difficult to solve optimally. However, for the case of one single RX, we show the existence of optimal rank-one solution to the SDR of the formulated QCQP, and thus solve the problem optimally.
For the general case with multiple RXs, we derive a new upper bound on the rank of the optimal SDR solution.
Based on the obtained SDR solution with higher rank, two approximate solutions are proposed by applying the techniques of TS and randomization,  respectively.
Furthermore, an efficient method to estimate the magnetic MIMO channel is proposed for the practical implementation of magnetic beamforming. Numerical results show the effectiveness of the proposed magnetic channel estimation and beamforming schemes as well as their great potential to significantly enhance the energy efficiency, maximum deliverable power, as well as performance fairness in multi-user MIMO MRC-WPT systems over the benchmark uncoordinated equal-current transmission.
\textcolor{black}{As a concluding remark, we would like to point out that in this paper, we have assumed that the RXs are well separated from each other, and thus the mutual inductances between them are negligible and thus ignored.  However, if the RX coupling is considered, our proposed circuit analysis and magnetic beamforming design need to be modified accordingly, which is worthy of investigation in future work.}

\begin{figure}[t!]
	\centering
	\includegraphics[width=.95\columnwidth]{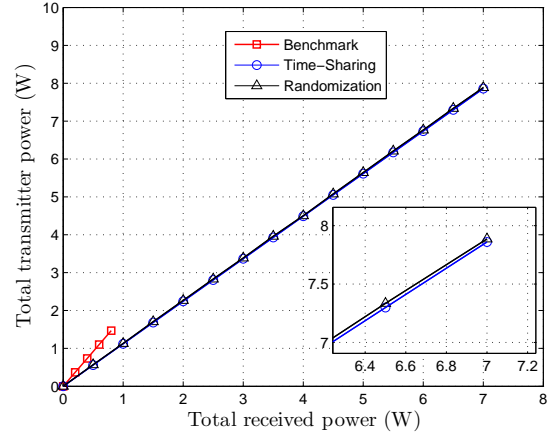}
	\caption{Total TX power v.s. total RX power for the four-user case.}
	\label{fig:Fig_opt_subopt}
	\vspace{-6mm}
\end{figure}

\appendices
\section{Proof to Theorem~\ref{the:solution_specialP2_Q1}} \label{sec:proof_lemma2}
For $Q=1$, we construct the Lagrangian of $\mathrm{(P2)}$ as
\begin{align}
  L(\bar{\bi}, v) &= \frac{1}{2} {\bar{\bi}}^H  {\overline{\bB}} {\bar{\bi}} + v \left( \alpha_1 P - \frac{w^2 }{2 r_{\textrm{rx}, 1}} \bar{\bi}^H \bM_1 \bar{\bi}\right). \label{eq:Lagrangian}
\end{align}
Then, the (Lagrange) dual function is given by
\begin{align}
L(v)&=  \alpha_1 P v \!+\! \underset{{\bar{\bi}}}{\inf} \;\frac{1}{2} {\bar{\bi}}^H  \! \left( \bR \!+\! \frac{w^2 (1 \!-\! v )}{r_{\textrm{rx}, 1}} \bM_1 \right) \bar{\bi}. \label{eq:LagrangeDual_1}
\end{align}
To obtain the best lower bound on the optimal objective value of ($\mathrm{P2}$), the dual variable $v$ should be optimized over $v \geq 0$ to maximize the dual function given in~\eqref{eq:LagrangeDual_1}. For dual feasibility, the dual function~\eqref{eq:LagrangeDual_1} should be bounded below. For convenience, we consider the following eigenvalue decomposition (EVD):
\begin{align}
 \bR + \frac{w^2 (1-v )}{r_{\textrm{rx}, 1}} \bM_1 = \bU \bPsi \bU^H, \label{eq:SVD_specialP3}
\end{align}
where the matrix $\bU = [\bu_1 \; \bu_2 \; \ldots \; \bu_N]$ is orthogonal, and $\bPsi = \diag\{\psi_1, \ldots, \psi_N\}$, with $\psi_1 \leq \psi_2 \leq \ldots \leq \psi_N$. For the case of arbitrary transmitter resistance values, the Lagrangian in \eqref{eq:Lagrangian} is bounded below in $\bar{\bi}$ and the dual function \eqref{eq:LagrangeDual_1} is maximized, only when $v$ is chosen as $v^{\star} >0$ such that $\psi_1 =0$.
Moreover, we observe that the objective in \eqref{eq:rewardP2} is minimized when the constraint \eqref{eq:const1P2} holds with equality, since both $\overline{\bB}$ and the matrix $\boldm \boldm^H$ are PSD. Hence, the optimal current can be written as ${\bar{\bi}}^{\star} = \beta \bu_1$, where $\bu_1$ is the eigenvector associated with the eigenvalue $\psi_1 =0$, and $\beta$ is a constant such that the constraint \eqref{eq:const1P2} holds with equality. Since the dual optimal solution leads to a primal feasible solution and the problem satisfies the Slater's condition~\cite{CVXTool}, the duality gap for $\textrm{(P2)}$ in the case of $Q=1$ is zero (although the problem is non-convex due to its non-convex constraints.)

For the special case of identical transmitter resistance, i.e., $\bR = r \bI_N$,
from the isometric property of the identity matrix $\bI_N$, the diagonal matrix $\bLambda$ is given by
\begin{align}
\bLambda &= \diag \left\{ r + \frac{w^2 (1-v )}{r_{\textrm{rx}, 1}}, \; r,\; \ldots, \; r\right\},
\nonumber
\end{align}
and the eigenvector $\bu_1 = \frac{\boldm_1}{\| \boldm_1 \|_2}$, and $\bu_n, \; \forall n \geq 2$, are arbitrarily orthogonal vectors constructed by methods such as Gram$-$Schmidt method. {It is easy to show that the Lagrangian in \eqref{eq:Lagrangian} is bounded below in $\bar{\bi}$ and the dual function \eqref{eq:LagrangeDual_1} is maximized, only when $v$ is chosen such that the first eigenvalue is zero}, i.e., the optimal dual variable is
\begin{align}
v^{\star} = 1 + \frac{r r_{\textrm{rx}, 1} }{w^2},
\end{align}
and the optimal current is thus given in~\eqref{eq:solution_wo_constraint_common_r}. The proof of Theorem~\ref{the:solution_specialP2_Q1} is thus completed.
\section{Proof to Theorem~\ref{the:rank-bound}} \label{sec:proof_theorem1}
Let $\bm{\lambda} = [\lambda_1 \; \ldots \; \lambda_Q]^T \geq 0, \bm{\rho}=[\rho_1 \; \ldots \; \rho_N]^T \geq 0$, and $\bm{\mu} =[\mu_1 \; \ldots \; \mu_N]^T$ $\geq 0$ be the dual variables corresponding to the constraint(s) given in~\eqref{eq:const1P1SDP},~\eqref{eq:const2P1SDP}, and~\eqref{eq:const3P1SDP}, respectively. Let the matrix $\bS  \succcurlyeq 0$ be the dual variable corresponding to the constraint $\bX \succcurlyeq 0$ in~\eqref{eq:const4P1SDP}. The Lagrangian of $(\mathrm{P1}\mathrm{-SDR})$ is then written as
\begin{align}
  L(\bX, \lambda, \bm{\rho}, \bS)  &= \frac{1}{2}   \Tr \left( {\overline{\bB}}  \bX \right) -  \\
  &\quad \sum \limits_{q=1}^Q \lambda_q \left( \Tr \left( {\bM_q} \bX \right) - \frac{2 r_q^2 \alpha_q P}{w^2 r_{ \textrm{l}, q} }\right)  + \nonumber \\
  & \quad  \sum \limits_{n=1}^N  \rho_n  \left( \Tr \left( {\bB_n} \bX \right) -A_{ n}^2 \right) + \nonumber \\
  &\quad \sum \limits_{n=1}^N  \mu_n   \left( \Tr \left( {\bW_n} \bX \right) - D_{ n}^2 \right) - \Tr \left( \bS \bX \right).
  \nonumber
\end{align} \vspace{-2mm}

Let $\bX^{\star}, \bm{\lambda}^{\star}, \bm{\rho}^{\star}, {\bm{\mu}^{\star}}$, and $\bS^{\star}$ be the optimal primal and dual variables, respectively. Since $\textrm{(P1-SDR)}$ is convex and satisfies the Slater's condition, the strong duality holds for this problem \cite{ConvecOptBoyd04}; as a result, the optimal primal and dual solutions should satisfy the Karush$-$Kuhn$-$Tucker (KKT) conditions given by \vspace{-2mm}
\begin{align}
  \nabla_{\bX}  L(\bX^{\star}, \bm{\lambda}^{\star}, \bm{\rho}^{\star}, \bm{\mu}^{\star}, \bS^{\star}) &= \frac{1}{2} {\overline{\bB}} - \sum \limits_{q=1}^Q \lambda_q^{\star} \bM_q + \nonumber \\
  &\quad \sum \limits_{n=1}^N \rho_n^{\star} {\bB_n} + \sum \limits_{n=1}^N \mu_n^{\star} {\bW_n}   - \bS^{\star} \nonumber \\
  &= \bm{0}. \label{eq:KKT_orthgonality} \\
  \bS^{\star} \bX^{\star} &= \bm{0}. \label{eq:KKT_complementarity}
\end{align} \vspace{-5mm}

Next, by multiplying \eqref{eq:KKT_orthgonality} by $\bX^{\star}$ on both sides and substituting \eqref{eq:KKT_complementarity} into the obtained equation, we have
\begin{align}
  \frac{1}{2} {\overline{\bB}} \bX^{\star} \!-\! \sum \limits_{q=1}^Q \lambda_q^{\star} \bM_q \bX^{\star} \!+\!  \sum \limits_{n=1}^N \rho_n^{\star} {\bB_n} \bX^{\star} \!+\!  \sum \limits_{n=1}^N \mu_n^{\star} {\bW_n} \bX^{\star}= \bm{0}. \label{eq:KKT_orthgonality2}
\end{align}
We thus have
\begin{align}
\rank &\left(\! \left( \! \frac{1}{2} {\overline{\bB}}+  \sum \limits_{n=1}^N \rho_n^{\star} {\bB_n}  +  \sum \limits_{n=1}^N \mu_n^{\star} {\bW_n} \! \right) \bX^{\star} \! \right) \nonumber \\
& =\rank \left( \! \sum \limits_{q=1}^Q \bM_q \bX^{\star} \! \right) \leq \rank \left( \! \sum \limits_{q=1}^Q  \bM_q \! \right) \leq Q. \label{eq:KKT_orthgonality3}
\end{align}
Since $\overline{\bB}$ is PSD, the matrix $\left( \frac{1}{2} {\overline{\bB}} +  \sum \limits_{n=1}^N \rho_n^{\star} {\bB_n} +  \sum \limits_{n=1}^N \mu_n^{\star} {\bW_n}  \right)$ must have full rank. Hence, \eqref{eq:KKT_orthgonality3} implies
\begin{align}
    \rank &\left( \bX^{\star} \right) \label{eq:KKT_orthgonality4}  \\
    &=  \rank   \left(\!  \left( \! \frac{1}{2} {\overline{\bB}} +  \sum \limits_{n=1}^N \rho_n^{\star} {\bB_n} +  \sum \limits_{n=1}^N \mu_n^{\star} {\bW_n} \! \right) \bX^{\star}  \! \right) \leq Q. \nonumber
\end{align}
On the other hand, from~\cite[Thm. 3.2]{HuangPalomar10}, we have
\begin{align}
  \rank  \left( \bX^{\star} \right) \leq \sqrt{Q+2N}. \label{eq:boundThm3.2}
\end{align}
Hence, from~\eqref{eq:KKT_orthgonality4} and~\eqref{eq:boundThm3.2}, the rank of the optimal solution $\bX^{\star}$ is upper-bounded as in~\eqref{eq:rank_bound}. The proof of Theorem 2 is thus completed.
\section{Proof of Theorem~\ref{the:LS}}\label{sec:proof_MCE}
With an estimate $\hatbM$, the squared error is given by
\begin{align}\label{eq:square_error}
  J(\hatbM) = \Tr \left( \left( \bG - \hatbM \tilbZ \right) \left( \bG - \hatbM \tilbZ \right)^H\right).
\end{align}
The LS estimate of $\bM$ is obtained by solving the following squared-error minimization problem,
\begin{align}\label{eq:LS_general}
  \hatbM = \underset{\hatbM}{\arg \min} \ \Tr \left( \left( \bG - \hatbM \tilbZ \right) \left( \bG - \hatbM \tilbZ \right)^H\right).
\end{align}
The derivative of the squared error $J(\hatbM)$ \textcolor{black}{with respect to $\hatbM$} is derived as
\begin{align}\label{eq:derivative}
  \frac{\partial J(\hatbM)}{\partial \hatbM}
  &= - \bG \tilbZ^H - \bG^{\ast} \tilbZ^T + \hatbM \left( \tilbZ \tilbZ^H + \tilbZ^{\ast} \tilbZ^T \right).
\end{align}

Since the $Q \times T$ RX-current matrix $\bZ$ has a full rank of $Q$ with probability one and the error matrix $\bE$ is random, the RX current matrix $\tilbZ= \bZ + \bE$ known at the central controller should also have a full rank of $Q$ with probability one and thus its inverse exists.
By setting the derivative in~\eqref{eq:derivative} to zero, the LS estimate is obtained as in~\eqref{eq:LS_est}. From~\eqref{eq:square_error} and~\eqref{eq:LS_est}, the corresponding least squared error is obtained as
\begin{align}
  \hatbM_{\textrm{LS}} = \left( \bG \tilbZ^H  + \bG^{\ast} \tilbZ^T \right) \left( \tilbZ \tilbZ^H  + \tilbZ^{\ast} \tilbZ^T \right)^{-1}.
\end{align}
The proof of Theorem~\ref{the:LS} is thus completed.
\bibliography{IEEEabrv,reference_170130}
\bibliographystyle{IEEEtran}

\end{spacing}
\end{document}